\documentclass[a4paper,10pt]{article}
  
\usepackage{amsthm,amsmath,amssymb,amscd,verbatim,graphicx,setspace,soul,hyperref,pifont}
\usepackage[dvipsnames]{xcolor}
\usepackage{centernot} 
\usepackage{bm}  
\usepackage{enumitem}
\setlist[enumerate]{itemsep=0mm}
\usepackage{mathtools}
\usepackage{soul} 
\usepackage{tikz-cd}

\usepackage[utf8]{inputenc}
\usepackage[english]{babel}
\usepackage{csquotes} 

\usepackage{geometry}
\usepackage{marginnote}

\hypersetup{colorlinks=true,citecolor=cyan,urlcolor=cyan}

\newcommand{\Set}[2]{%
  \{#1 : #2 \}%
} 

\usepackage[round,longnamesfirst]{natbib}
\let\originalleft\left
\let\originalright\right
\renewcommand{\left}{\mathopen{}\mathclose\bgroup\originalleft}
\renewcommand{\right}{\aftergroup\egroup\originalright}

\newcommand\Graph{\Gamma_{\!\se}}
\newcommand\sGraph{\Gamma_{\!\s}}
\newcommand\ArchCl{\Pi_{\se}}
\newcommand\ArchCls{\Pi_{\s}}

\def\AA{\mathcal{A}} 
 \def\CC{\mathcal{C}}
  \def\FF{\mathcal{F}}
  
 \def\KK{\mathcal{K}} \def\LL{\mathcal{L}}
   
\def\PP{\mathcal{P}}  	

 \def\UU{\mathcal{U}}

  \def\NNN{\mathbb{N}} 
   
\def\RRR{\mathbb{R}}

\def\<{\left<} \def\>{\right>} 
\def\({\left(} \def\){\right)}

\DeclareMathOperator\rInt{ri}

\DeclareMathOperator\Span{span}
\DeclareMathOperator\supp{supp}

\theoremstyle{plain}
\newtheorem{thm}{Theorem}[section]
\newtheorem{prop}[thm]{Proposition}

\newtheorem{cor}[thm]{Corollary}
\newtheorem{lemma}[thm]{Lemma}

\theoremstyle{definition}

\newtheorem{example}[thm]{Example}
\newtheorem*{original standing assumptions}{Original Harsanyi Assumptions}
\newtheorem*{standing assumptions}{Standing assumptions}
\newtheorem{rem}[thm]{Remark}

\newcommand\ov[1]{\overline{#1}}

\def\inv{^{-1}}

\def\Geq{\succsim}

\def\sGeq{\succ}

\def\Eq{\sim}

\def\nsGeq{\not\sGeq}

\def\GeqV{\succsim_V}

\def\sGeqV{\succ_{V}}
\def\EqV{\sim_V}

\def\aa{\alpha}
\def\bb{\beta}

\def\ee{\epsilon}
\def\kk{\kappa}
\def\ll{\lambda}
\def\mm{\mu}

\newcommand{\restr}[1]{\lower.45ex\hbox{$|$}\lower.5ex\hbox{}_{#1}} 

\DeclareMathOperator{\aff}{aff} 			
\DeclareMathOperator{\rai}{rai} 			
\DeclareMathOperator{\cone}{cone} 		

\makeatletter
\newcommand*{\textoverline}[1]{$\overline{\hbox{#1}}\m@th$}
\makeatother

\title{Expected utility theory on mixture spaces without the completeness axiom%
\thanks{We are grateful for extremely helpful advice from two anonymous reviewers.  We also thank Efe Ok for invaluable discussion of the history of the subject.  
David McCarthy thanks the Research Grants Council of the Hong Kong Special Administrative Region, China (HKU 750012H) for support. Teruji Thomas thanks the Leverhulme trust for funding through the project `Population Ethics: Theory and Practice' (RPG-2014-064). This is a heavily revised version of a preprint `Representation of strongly independent preorders by sets of scalar-valued functions' MPRA. Paper No. 79284 (2017). Declarations of interest: none.}
}
\author{David McCarthy%
\footnote{Corresponding author, Dept. of Philosophy, University of Hong Kong, Hong Kong, 
\href{mailto:mccarthy@hku.hk}{mccarthy@hku.hk}} 
\and 
Kalle Mikkola%
\footnote{Dept. of Mathematics and Systems Analysis, Aalto University, Finland,
\href{mailto:kalle.mikkola@iki.fi}{kalle.mikkola@iki.fi}} 
\and 
Teruji Thomas%
\footnote{Global Priorities Institute, University of Oxford, United Kingdom, 
\href{mailto:teru.thomas@oxon.org}{teru.thomas@oxon.org}}
}

\def\se{\Geq}

\def\e{\Eq}

\def\s{\sGeq}

\def\seV{\GeqV}
\def\sV{\sGeqV}
\def\eV{\EqV}

\numberwithin{equation}{section}

\begin{document}
\date{\vspace{-5ex}}

\date\today

\maketitle







\begin{abstract}
\noindent A mixture preorder is a preorder on a mixture space (such as a convex set) that is compatible with the mixing operation. In decision theoretic terms, it
satisfies the central expected utility axiom of strong independence. We consider when a mixture preorder has a multi-representation that consists of real-valued, mixture-preserving functions. If it does, it must satisfy the mixture continuity axiom of \citet{HM1953}. 
Mixture continuity is sufficient for a mixture-preserving multi-representation when the dimension of the mixture space is countable, but not when it is uncountable.   
Our strongest positive result is that mixture continuity is sufficient in conjunction with a novel axiom we call countable domination, which constrains the order complexity of the mixture preorder in terms of its Archimedean structure. 
We also consider what happens when   the mixture space is given its natural weak topology. Continuity (having closed upper and lower sets) and closedness (having a closed graph) are stronger than mixture continuity. We show that continuity is necessary but not sufficient for a mixture preorder to have a mixture-preserving multi-representation. Closedness is also necessary; we leave it as an open question whether it is sufficient. We end with results concerning the existence of mixture-preserving multi-representations that consist entirely of strictly increasing functions, and a  uniqueness result.


\end{abstract}

\section{Introduction}\label{s:intro}

The importance of allowing for incomplete preferences is by now beyond dispute. In the context of expected utility, \citet[p. 630]{vNM1953} themselves remarked, of the completeness axiom, that it is ``very dubious, whether the idealization of reality which treats this postulate as a valid one, is appropriate or even convenient''. In the first systematic treatment of expected utility without the completeness axiom, \citet[p. 446]{rA1962} wrote that while all the expected utility axioms are descriptively implausible,  the completeness axiom alone is ``hard to accept even from the normative viewpoint''. 
With normative questions 
especially in mind, we address the problem of representing incomplete preferences
by sets of utility functions. 

Following \citet{rA1962} and \citet{SB1998},  we suppose that preferences are given by a preorder on a mixture space, in the sense of \citet{mH1954}. 
A mixture space is a set $M$ together with a mixing operation, so that for any elements $x$ and $y$ in $M$ and $\aa \in [0,1]$, the element $x \aa y$ of $M$ is understood to be a mixture of $x$ and $y$ in which $x$ is given weight $\aa$ and $y$ weight $1-\aa$. 
We give the standard axiomatization of mixture spaces in section~\ref{s:main}. For now, the best known example involving uncertainty is when $M$ is the set of 
probability measures
on some outcome space, and $x \aa y$ is taken to be the 
probability measure
$\aa x + (1-\aa)y$. 
More generally, any convex set, and thus any vector space, is a mixture space, with the mixing operation defined by the same formula.

Given a possibly incomplete preorder $\se$ on mixture space $M$, a \emph{multi-representation} is a 
nonempty set $\UU$ of functions $M\to \RRR$ such that $x\se y$ if and only if, for all $u\in\UU$, $u(x)\geq u(y)$.%
 \footnote{The concept of a multi-representation of a preorder was introduced in \citet{eO2002}, but the general idea goes back much further. In decision theory, \citet{tB1986} is perhaps the earliest explicit example, but in the guise of a single vector-valued function, rather than a family of scalar-valued functions, multi-representations were envisioned but not developed in \citet[pp. 19--20]{vNM1953}. There is no reason, however, why the general concept of a multi-representation has to stipulate that the codomain is the real numbers. For an example in which it is taken to be a linearly ordered abelian group, see \citet{mP2013}.}

It is natural to require that the functions $u$ respect the mixing operation. 
A function $u \colon M \to M'$ between mixture spaces is {\em mixture preserving} when $u(x \aa y) = u(x)\aa u(y)$. 
In a multi-representation, as we have defined it, 
$M'$ 
is the vector space of 
real numbers. So the question we consider is 
under what conditions a preorder $\se$ on $M$ has a {\em mixture-preserving multi-representation}; that is, under what conditions does it satisfy 

\newlist{MR}{enumerate}{1}
\setlist[MR]{label=\textbf{MR},ref={\text{MR}}}

\begin{MR}[align=parleft, leftmargin=!,itemsep=0pt,labelsep=14pt] 
\item\label{MR} There is a nonempty 
set $\UU$ of 
mixture-preserving functions $M\to\RRR$, such that for all $x,y\in M$, 
\[
	x \se y \iff u(x) \geq u(y) \text{ for all } u \in \UU.
\]
\end{MR}

It is well known that \emph{any} mixture space is isomorphic to a convex set.  Using this fact,  our question is mathematically equivalent to the question of when a preorder on a  convex set has a multi-representation  consisting of 
affine (or even linear) 
functionals on the ambient vector space, 
restricted to the convex set. 
We will exploit this equivalence in proofs (see section~\ref{s:mspace}), but we follow \citet{pM2001} in thinking that mixture spaces are conceptually more fundamental for decision theory. For example, it is often easier to verify that an algebraic structure of interest to decision theorists is a mixture space  than to show directly that it is isomorphic to a convex set.

Much of the literature on mixture-preserving multi-representations has focussed on specific types of mixture spaces. 
Besides sets of probability measures (with different possible assumptions about the underlying measurable space), 
examples include sets of Savage-acts, at least given mild structural assumptions \citep{GMMS2003}; Anscombe-Aumman acts; charges (i.e. finitely additive measures); and vector-valued measures representing imprecise probabilities. Mixture-preserving multi-representations themselves come in a variety of forms. 
In the popular Anscombe-Aumman setting, for example, 
incomplete preferences may be a matter of incomplete beliefs, incomplete tastes, or both, and multi-representations can reflect these distinctions.%
\footnote{For examples involving incomplete beliefs, see \citet{tB1986, tB2002, GMMS2003}. For tastes, see \citet{DMO2004, EO2006, oE2008, oE2014, lG2017, HOR2019, dB2020}. For beliefs and tastes, see \citet{SSK1995, rN2006, OOR2012}. For closely related examples, see \citet{MM2008} (interval-valued representations), \citet{GK2012, GK2013} (nonstandard preorders), and  \citet{yH2012} (justifiable choice).}  



While one could consider these different frameworks one at a time, taking into account their special features, we think it is interesting to consider the unifying question of when one may obtain a mixture-preserving multi-representation of a preorder on an abstract mixture space. This fits with the appealing methodology of assuming as little mathematical structure as possible, and addressing general questions with general tools.


To introduce our main results, let us mention two 
axioms that must clearly be satisfied for \ref{MR} to hold, i.e.~for the existence of a mixture-preserving multi-representation.%
\footnote{We define these axioms in section \ref{s:main}.  Slightly different versions of the two axioms are common in the literature; we clarify some of the relationships in appendix \ref{s:WCon}.} 
First, the preorder must be 
what we call a `mixture preorder': it must satisfy what is arguably the central axiom of expected utility theory, strong independence. 
Strong independence is not in general a natural assumption for preferences on mixture spaces; few people's preferences satisfy it on the simplex whose points denote different proportions of coffee, milk, and sugar.  But it is a plausible normative requirement in the examples of mixture spaces introduced above, which all involve uncertainty. 
Second, it is not hard to show that if a mixture preorder has a mixture-preserving multi-representation, it must satisfy the mixture continuity axiom of  \citet{HM1953}.

The result which sets the stage for our discussion, Theorem~\ref{t:MC not suff}, shows, we think rather surprisingly, that  mixture continuity is not sufficient for a mixture preorder to have a mixture-preserving multi-representation. However, mixture preorders that satisfy mixture continuity without having a mixture-preserving multi-representation must be rather complicated; for example, Theorem~\ref{t:count} shows that  they must have uncountably infinite dimension. 
This raises the question of whether there are normatively natural ways of strengthening or supplementing mixture continuity that do guarantee 
\ref{MR}.

Our strongest positive result, Theorem~\ref{t:CD}, shows that, in combination with mixture continuity, an axiom we call `countable domination' is sufficient for
a mixture preorder to satisfy \ref{MR}.
We provide two interpretations of this axiom. 
First, it is a member of a natural but apparently novel family of decision-theoretic axioms that constrain what we call the `Archimedean structure' of the preorder. 
Another axiom in this family is the standard Archimedean axiom, which is much stronger than countable domination. 
Second, countable domination may be seen as a dimensional restriction on mixture preorders that is much 
less demanding than the requirement of countable dimension.

Our strongest negative result, Theorem~\ref{t:top}, 
considers what happens if we impose a topology on mixture spaces and upgrade mixture continuity to a stronger continuity condition. 
It notes that any mixture preorder that 
satisfies \ref{MR} must be both continuous and closed in the weak topology, understood as the coarsest topology on the mixture space in which the 
real-valued 
mixture-preserving functions 
are continuous. 
However, more surprisingly, it also shows that being continuous is not sufficient for \ref{MR}. 
We leave it as an open question whether being closed is sufficient. 

Section~\ref{s:main} states our axioms more formally and presents our main results. Section~\ref{s:related lit} relates them to the most immediately relevant literature, 
showing how they extend results of \cite{SB1998} 
and answer a question posed by \citet{DMO2004}. 
Section~\ref{s:CD} discusses the interpretation of countable domination. Section~\ref{s:discuss} 
provides proofs of our main results; 
it emphasizes the central ideas, 
appealing to a series of auxiliary results whose proofs we defer to appendix \ref{s:aux proofs}. Section~\ref{s:refine} refines our results by considering two topics. Section~\ref{s:strict} presents results concerning the existence of mixture-preserving multi-representations that 
consist entirely of strictly increasing functions,    
and relates them (in section~\ref{s:strict related lit}) to results by \citet{rA1962, DMO2004, oE2014} and \citet{lG2017}. 
Section~\ref{s:unique} presents a uniqueness result for mixture-preserving multi-representations that is an abstract version of the uniqueness result of \citet{DMO2004}. Appendix \ref{s:WCon} explains the connection between our independence and mixture-continuity axioms and slightly different ones common in the literature. Appendix \ref{s:weak dom} provides a geometrical interpretation of our discussion of Archimedean structures. And, as we mentioned, appendix \ref{s:aux proofs} contains proofs of the auxiliary results.

Finally, we acknowledge the centrality to our results of the work of \citet{vK1953}.

%

\section{Main results}\label{s:main}
 A {\em mixture space} is a nonempty set $M$ together with a 
{\em mixing operation} $m \colon M \times M \times [0,1] \to M$ that satisfies 
axioms shortly to be described. 
As is customary, when the mixing operation is understood, 
we write $x\aa y$ for $m(x,y,\aa)$. 
The axioms are then: 
(i) $x\aa y=y(1-\aa)x$; (ii) $x\aa x=x$; (iii) if $x\aa z = y\aa z$ for some $\aa \neq 0$, then $x=y$; and (iv) $x\aa(y\bb z) = (x \tfrac{\aa}{\aa+\bb-\aa\bb}y)(\aa+\bb-\aa\bb)z$ if $\aa$ and $\bb$ are not both zero.%
\footnote{These are a reordering of the axioms given by \citet{mH1954}. Mixture {\em sets}, as used for expected utility theory in e.g. \citet{HM1953} and \citet{pF1970, pF1982} are more general. Terminology varies; \citet{pM2001} uses `non-degenerate mixture sets' for what we are calling mixture spaces. In our terminology, despite the greater generality of mixture sets, Mongin recommends focussing on mixture spaces for the development of decision theory.} 
These axioms abstract features of convex subsets of vector spaces, where the mixing operation is given by $x \aa y = \aa x + (1-\aa)y$. The first three are self-explanatory, and the last is an associativity axiom. 

We will need the notion of the dimension of a mixture space. The standard definition \citep{mH1954} reduces to the case of convex sets (see section \ref{s:mspace}). However, it is more in the spirit of our focus on mixture spaces to provide a characterisation directly in terms of the mixture-space structure. 
Given a mixture space $M$, say that $M' \subset M$ is a {\em mixture subspace} of $M$ if it is a mixture space under the mixing operation inherited from $M$. 
For any nonempty $A \subset M$, let $M(A)$ be the smallest mixture subspace of $M$ containing $A$. 
Say that $A$ is \emph{mixture independent} if, for any 
nonempty 
$A_1,A_2\subset A$, $A_1\cap A_2=\emptyset\implies M(A_1)\cap M(A_2)= \emptyset$.%
\footnote{This is analogous to the following characterisation of linear independence of a subset $B$ of a vector space: for any $B_1,B_2\subset B$, $B_1\cap
B_2=\emptyset \implies \Span(B_1)\cap\Span(B_2)=\{0\}$.}
We define the {\em dimension} of $M$,
written $\dim M$, 
to be $|A|-1$ for any maximal mixture-independent subset $A$. 
In section~\ref{s:prelim} we show this is well defined and equivalent to the customary definition.

A {\em mixture preorder} is a preorder $\se$ on a mixture space $M$ that is compatible with the mixing operation in that it satisfies the following axiom:
\newlist{SI}{enumerate}{1}
\setlist[SI]{label=\textbf{SI},ref={\text{SI}}}

\begin{SI}[align=parleft, leftmargin=!,itemsep=0pt,labelsep=14pt] 
        \item\label{SI} For $x$, $y$, $z \in M$, and $\aa\in (0,1)$,
$
x\se y \iff x \aa z \se y \aa z.
$   
\end{SI}
\noindent A {\em preordered mixture space} is a pair $(M, \se)$ where $M$ is a mixture space and $\se$ is a mixture  preorder on $M$. 
When $M$ is a convex set of probability measures, \ref{SI} is {\em strong independence}, arguably the central axiom of expected utility theory.


We are interested in the question: when does a mixture preorder have a mixture-preserving multi-representation?

Consider the following axiom, introduced by \citet{rA1962}.%
\footnote{\label{n:aumann}However, \citet{rA1962} regarded \ref{MC} as too strong for his purposes, and instead focussed on, in our labelling: 

\newlist{Au}{enumerate}{1}
\setlist[Au]{label=\textbf{Au},ref={\text{Au}}}

\begin{Au}[align=parleft, leftmargin=!,itemsep=0pt,labelsep=14pt] 
        \item\label{Au} For $x$, $y$, $z \in M$, if $x \aa y \s z  \text{ for all } \aa\in (0,1]$, then 
 $z \nsGeq y$.
\end{Au}
This axiom is strong enough to rule out, for example, the lexicographic ordering of the unit square. But as well as being weaker than \ref{MC}, 
for mixture preorders, \ref{Au} is also weaker than the axiom \ref{Ar} discussed below.  
We discuss \ref{Au} further in 
section~\ref{s:strict related lit}.}

\newlist{MC}{enumerate}{1}
\setlist[MC]{label=\textbf{MC},ref={\text{MC}}}

\begin{MC}[align=parleft, leftmargin=!,itemsep=0pt,labelsep=14pt] 
        \item\label{MC} For $x$, $y$, $z \in M$, if $x \aa y \s z \text{ for all } \aa\in (0,1]$, then $y \se z$.
\end{MC}

As Aumann noted, 
for mixture preorders, \ref{MC} is 
equivalent to 
the well-known {\em mixture continuity} axiom of \citet{HM1953}, 
that $\Set{\aa \in [0,1]}{x \aa y \se z}$ 
and $\Set{\aa \in [0,1]}{z \se x\aa y}$ 
are closed in $[0,1]$ for all $x,y,z\in M$.%
\footnote{See section \ref{s:related lit} and appendix \ref{s:WCon}  for further clarification of the connection between \ref{MC}, the Herstein-Milnor axiom \ref{HM}, and the related axiom \ref{WCon} used by \citet{SB1998} and \citet{DMO2004}. In particular, we explain in Remark \ref{r:MC} why they are all equivalent for mixture preorders.}

Our interest in the axiom \ref{MC} is prompted by the trivial observation, recorded in the following, that \ref{MC} is necessary for \ref{MR}. However, to our surprise, \ref{MC} is not sufficient:
\begin{thm}\label{t:MC not suff}
For any preordered mixture space $(M, \se)$, 
\[
	\ref{MR} \implies \ref{MC},
\]
but the implication is not reversible. 
\end{thm}
The failure of reversibility is in fact quite general.
\begin{thm}\label{t:klee}
Every mixture space of uncountable dimension has a mixture preorder that satisfies \ref{MC} but violates \ref{MR}. 
\end{thm}

This raises the question: how might \ref{MC} be strengthened to guarantee a mixture-preserving multi-representation? 
We will consider a range of conditions that are stronger than \ref{MC}. Some we will show are sufficient for a mixture-preserving multi-representation, but not necessary. Others are necessary, but not sufficient. 
We do not know of a 
nontrivial 
condition that is necessary and sufficient, 
but one of our results will suggest a natural candidate.%
\footnote{\label{n:quotients}We note in passing that, if $(M,\se)$ is a preordered mixture space, then the quotient $M/{\e}$ is also naturally a mixture space with a mixture preorder $\se'$, and $\se'$ is actually a partial order ($x\e' y\implies x=y$).  For many purposes it suffices to consider $M/{\e}$ rather than $M$. In particular, it is not hard to see that $\se$ satisfies \ref{MR} if and only if $\se'$ does. But we will focus on $M$ itself.}

A first sufficient condition for \ref{MR} is suggested by 
Theorem~\ref{t:klee}: we simply strengthen \ref{MC} by assuming in addition that $\dim M$ is countable.
 (Recall that \emph{countable} means either finite or countably infinite.)

\begin{thm}\label{t:count}
For any preordered mixture space $(M, \se)$, 
\[
	\ref{MC} \,\&\,\dim M \text{ is countable} \implies \ref{MR},
\]
but the implication is not reversible.
\end{thm}
However, the assumption of countable dimension is clearly much stronger than necessary. 
We will give some examples in section \ref{s:CD}: in particular, Example~\ref{ex:combo} provides two simple ways in which a preordered mixture space of countable dimension that satisfies \ref{MC}, and consequently \ref{MR}, can be blown up to one of arbitrarily large dimension that still satisfies both \ref{MC} and \ref{MR}. 

Instead, our weakest sufficient condition
involves an apparently novel 
axiom that we call \emph{countable domination} (\ref{CDom}). 
We state it now but will discuss its significance at length in section \ref{s:CD}; in short,  it strictly weakens the assumption that $\dim M$ is countable, and can also be seen as a much weaker form of the standard Archimedean axiom.

Let $\Graph\subset M\times M$ be the graph of the mixture preorder $\se$: it consists of pairs $(x,y)$ with $x\se y$. For any $(x,y)$ and $(s,t)$ in $\Graph$, say that $(x,y)$  \emph{weakly dominates} $(s,t)$ if $x\aa t\se y\aa s$ for some $\aa\in (0,1)$. 
The relation of weak domination is a preorder on $\Graph$ (see appendix~\ref{s:weak dom}). 
A natural interpretation is that when  $(x,y)$ weakly dominates $(s,t)$,  the (weakly positive) difference in value between $s$ and $t$ is at most finitely many times greater than that between $x$ and $y$. Our axiom is 
\newlist{CDom}{enumerate}{1}
\setlist[CDom]{label=\textbf{CD},ref={\text{CD}}}
\begin{CDom}[align=parleft, leftmargin=!,itemsep=0pt,labelsep=14pt] 
        \item\label{CDom}  There is a countable set $D \subset \Graph$ such that 
       each
        $(s,t)\in \Graph$ is weakly dominated by some $(x,y)\in D$. 
\end{CDom}

Our strongest positive result is

\begin{thm}\label{t:CD}
For any preordered mixture space $(M, \se)$,
\[
\ref{MC} \,\&\, \ref{CDom}  \implies \ref{MR},
\] 
but the implication is not reversible. 
\end{thm}


Instead of adding to \ref{MC} a condition such as \ref{CDom},  
we might impose a topology on the mixture space, and upgrade \ref{MC} to a stronger continuity condition.

Given an arbitrary topological space $M$, we say that a preorder $\se$ on $M$ is {\em continuous} if, for all $x\in M$, the sets  $\Set{y \in M}{y \se x}$ and $\Set{y \in M}{x \se y}$ are closed in $M$. A stronger continuity-like condition that is sometimes used is that the graph $\Graph$ is closed in the product topology on $M\times M$; in this case we simply say that $\se$ is {\em closed}.%
\footnote{In the study of arbitrary preorders on  topological spaces, the  distinction between these two forms of continuity is standard, but terminology varies. For example, \citet{EO2011} use `semicontinuous' and `continuous' for our `continuous' and `closed' respectively.
\citet{BH2016} use `semi-closed' and `closed'.} 
Thus we study the following axioms. 

\newlist{Con}{enumerate}{1}
\setlist[Con]{label=\textbf{Con},ref={\text{Con}}}

\begin{Con}[align=parleft, leftmargin=!,itemsep=0pt,labelsep=14pt] 
        \item\label{Con}  $\se$ is continuous.  
\end{Con}

\newlist{Cl}{enumerate}{1}
\setlist[Cl]{label=\textbf{Cl},ref={\text{Cl}}}

\begin{Cl}[align=parleft, leftmargin=!,itemsep=0pt,labelsep=14pt] 
        \item\label{Cl}  $\se$ is closed.  
\end{Cl}

Specific examples of mixture spaces (like sets of probability measures) may suggest specific topologies 
(see section~\ref{s:related lit}).   
However, we will focus on what we call the \emph{weak topology},  which makes sense for any mixture space. By definition, it is 
the coarsest topology (i.e. containing the fewest open sets) such that all the mixture-preserving functions $M \to \RRR$ are continuous. 
See Remark \ref{rem:weak} below for more on our terminology. The interest of the weak topology comes from the fact that it makes both \ref{Cl} and \ref{Con} into necessary conditions for \ref{MR}, as the following elaboration of Theorem~\ref{t:MC not suff} explains.

\begin{thm}\label{t:top}
For any preordered mixture space $(M, \se)$ in which $M$ has the weak topology, 
\[
\ref{MR}  \implies \ref{Cl} \implies \ref{Con} \implies \ref{MC}, 
\]
but the second and third implications are not reversible. 
\end{thm}
\noindent 
As before, the displayed implications are easily proved and essentially well known; the novelty lies in the failures of reversibility. 
In particular, Theorem \ref{t:top} shows that \ref{Con} is still not sufficient for \ref{MR}. This is our strongest negative result;  
it is somewhat delicate because  \ref{Con}, unlike \ref{MC}, \emph{does} entail \ref{MR} when,
for example, 
$M$ is a vector space  (see Remark~\ref{rem:klee}). 
For us, it is an open question whether \ref{Cl} and \ref{MR} are equivalent. 
Of course, by Theorem \ref{t:CD},   all four conditions are equivalent when \ref{CDom} holds.

\begin{rem}\label{rem:weak}
A vector space $V$ is a mixture space, so, as we have defined it, the weak topology on $V$ is the coarsest one that makes every 
mixture-preserving function $V \to \RRR$ continuous. This is equivalent to the more standard definition of the weak topology on a vector space as the coarsest one that makes every linear functional on $V$ continuous, since a function on $V$ is mixture preserving if and only if it is affine (i.e. linear plus a constant).

In the vector space case, there are, of course, a variety of weak topologies, each induced by a given subspace of linear functionals. Similarly, there are a variety of weak topologies on mixture spaces, corresponding to subspaces of mixture-preserving functions. But unless otherwise stated, we will not be discussing other weak topologies, hence our use of the term {\em the} weak topology. Other basic features of the weak topology on a mixture space are noted in Lemma~\ref{l:MS} in appendix~\ref{s:aux proofs}.
\end{rem}

Following discussion of our axiom \ref{CDom} in section~\ref{s:CD}, section~\ref{s:discuss} presents proofs of the above results, while relegating technical work to appendix \ref{s:aux proofs}. Section~\ref{s:refine} refines the picture in two ways. First, if $(M,\se)$ is a preordered mixture space, we say that a function $f\colon M\to \RRR$ is {\em increasing} if $x\se y$ implies $f(x)\geq f(y)$, and \emph{strictly increasing} if, in addition, $x\s y$ implies $f(x)>f(y)$. A mixture-preserving multi-representation clearly consists of functions that are increasing, but they need not be strictly increasing. Section~\ref{s:strict} gives results concerning the existence of mixture-preserving multi-representations that contain only strictly increasing functions. Second, section~\ref{s:unique} provides a uniqueness result for mixture-preserving multi-representations that is essentially an abstract version of the uniqueness result given by \citet{DMO2004}.

\subsection{Related literature}\label{s:related lit}
In section~\ref{s:intro} we noted the wide variety of  types of mixture spaces, and forms of mixture-preserving multi-representations, that have been discussed. While it would be desirable to consider whether our abstract results have applications in all of those areas, that project lies well beyond the scope of this article. Instead, we will first discuss how our results improve on those of \citet{SB1998}, and then present one application: we explain how one of our results solves a problem left open by the influential work of \citet{DMO2004}.

Our basic objects of study are preorders on mixture spaces that satisfy \ref{SI} and \ref{MC}. It is common---and is done so specifically by Shapley and Baucells, and Dubra {\em et al}---to focus on a slightly different set of basic axioms; we refer to these as `independence' (\ref{Ind}), which is strictly weaker than \ref{SI},  and `weak continuity' (\ref{WCon}), which is strictly stronger than \ref{MC}. However,  our axioms \ref{SI} and \ref{MC} are together equivalent to their axioms \ref{Ind} and \ref{WCon}. This equivalence seems to have been known already by Shapley and Baucells (see their note 1), but since formal statements and proofs are hard to find, we provide details in appendix~\ref{s:WCon}. 
 For ease of comparison, we take the liberty of presenting their results in terms of our axioms and terminology.

Shapley and Baucells used a standard embedding theorem to show that any mixture preorder is naturally associated with an essentially unique convex cone. We explain this technique, which we will also use, in section~\ref{s:mspace}. 
They called a mixture preorder `proper' if its cone has a nonempty relative algebraic interior; see section~\ref{s:t:CD} for the definition. 
Their main result 
on mixture-preserving multi-representations showed that every proper mixture preorder that satisfies \ref{MC} also satisfies \ref{MR}. 
As Shapley and Baucells observed, properness holds automatically when the mixture space is finite-dimensional.  
Thus they effectively proved a weaker version of our Theorem~\ref{t:count}, in which `countable' is replaced by `finite'. More importantly, our Theorem~\ref{t:CD} strengthens their main result, as our axiom \ref{CDom} is much weaker than their assumption of properness. 
Indeed, properness is equivalent to a strengthening of \ref{CDom} that we call `singleton domination' (\ref{SDom}), to be introduced in section \ref{s:CD}.

The assumption of properness was criticized by \citet[p. 127]{DMO2004}: ``Unfortunately, it is not at all easy to see what sort of a primitive axiom on a preference relation would support such a technical requirement.'' 
Our axioms \ref{CDom} and \ref{SDom} are not subject to this kind of criticism.  They are formulated directly in terms of the preorder, and, as we explain in section \ref{s:CD}, they are members of a  natural family of axioms that place limits on the complexity of the preorder in terms of its Archimedean structure. The standard Archimedean axiom is 
a much stronger 
axiom of this type.

\citet{DMO2004} consider the mixture space $M=P(X)$ of Borel probability measures on a compact metric space $X$. Let $C(X)$ be the set of continuous functions $X \to \RRR$. 
They endow $P(X)$ with the narrow topology 
(or what Dubra {\em et al} call the topology of weak convergence): 
the coarsest topology such that all the functions $P(X)\to\RRR$, defined by integrating against functions in $C(X)$, are continuous.%
\footnote{When $X$ is finite, the narrow topology is equal to what we have called the weak topology; when $X$ is infinite, it is more coarse, i.e. contains fewer open sets, strengthening \ref{Con} and \ref{Cl}. 
As well as by \citeauthor{DMO2004}, this strengthened form of \ref{Cl} is used in the context of multi-representations by e.g. \citet{GMMS2003, OOR2012} and \citet{lG2017}.} 
Their expected multi-utility theorem shows that \ref{Cl} is enough to ensure 
that any mixture preorder on $M$ has a mixture-preserving multi-representation that consists of expectational functions: 
functions of the form $p \mapsto \int_X u \, \mathrm{d} p$ for some $u \in C(X)$.%
\footnote{Their result contains more detail than this. For discussion and further elaboration, see \citet{oE2008} and \citet{HOR2019}.}
They raise the question of whether this result would hold if \ref{Cl} was weakened to \ref{Con} or \ref{MC}, 
noting only that \ref{MC} is enough when $X$ is a finite set.%
\footnote{\label{n:SB DMO}This follows from the 
result about finite dimensionality due to \citet{SB1998} noted above, since every $P(X)$ with $X$ finite is a finite-dimensional mixture space (of dimension $|X|-1$). 
The \citeauthor{SB1998} result is slightly stronger though, as 
not every finite-dimensional mixture space is isomorphic to some $P(X)$. For example, $(0,1)$ is a one-dimensional mixture space but it is not isomorphic to $P(\{0,1\})\cong [0,1]$.}
Our Theorem \ref{t:klee} shows that \ref{Cl} {\em cannot} be weakened to \ref{MC} in their expected multi-utility theorem, since when $X$ is infinite, $P(X)$ has uncountable dimension. 
We do not know whether \ref{Cl} can be weakened to \ref{Con} in their result, but Theorem \ref{t:top} shows that there can be no general inference from \ref{Con} to \ref{Cl}. 

There is large body of literature on the general question of when a preorder on an arbitrary topological space has a continuous multi-representation (a condition we call \ref{CMR}). 
In requiring a mixing-structure, along with mixture-preserving multi-representations, the focus of this article has been different. In the general setting, it is well-known that being closed is not sufficient for \ref{CMR}. One source of counterexamples is a topological vector space (and hence mixture space): $L^p[0,1]$, with the usual norm, with $0 < p < 1$, which has no non-zero continuous linear functionals \citep[\S 1.47]{R3}. As far as we know, the strongest necessary condition for 
\ref{CMR} to hold is given by \citet{BH2016}, under the assumption that the topology is first countable. 
Bosi and Herden remark that they do not see any possibility for satisfactorily avoiding that assumption. Turning back to our setting, the weakest sufficient condition we have for \ref{MR} to hold is the conjunction of \ref{MC} and another type of countability condition, \ref{CDom}. 
Despite the fact that \ref{CDom} is clearly a long way from necessary for \ref{MR}, we likewise do not see a satisfactory strategy for weakening it.

\section{Countable domination}\label{s:CD}
We now discuss our axiom \ref{CDom}, and provide some examples. First, we show that it is a natural weakening of the well-known Archimedean axiom, and connect it with the idea of Archimedean classes. Second, we explain how it 
 weakens the assumption that $M$ has countable dimension.

\subsection{Countable domination as a weak Archimedean axiom}\label{s:weak Arch}
To better understand \ref{CDom}, we now introduce two more axioms that are in the same natural class. As we will explain, the axioms in this class can be interpreted as constraining the order complexity of mixture preorders. 

Given a preordered mixture space $(M, \se)$, 
let $\sGraph \subset \Graph$ consists of pairs $(x,y)$  with $x\s y$. 
Our first axiom is the following. 
\newlist{Ar}{enumerate}{1}
\setlist[Ar]{label=\textbf{Ar},ref={\text{Ar}}}
\begin{Ar}[align=parleft, leftmargin=!,itemsep=0pt,labelsep=14pt] 
        \item\label{Ar} 
        Every $(x,y) \in \sGraph$ 
        weakly dominates every $(s,t)\in \Graph$.  
\end{Ar} 
Recall the {\em Archimedean axiom}, stated by \citet{vNM1953}: if $x \s y$ and $y \s z$, then  $x\aa z \s y$  and $y \s x\bb z$ for some $\aa$ and $\bb$ in $(0,1)$. 
It is straightforward to show that for mixture preorders, \ref{Ar} is equivalent to the Archimedean axiom. 

Our second axiom  is notable because of its close connection to  the approach of \citet{SB1998}; see  section~\ref{s:related lit}. 
We call this apparently novel axiom {\em singleton domination}.
\newlist{SDom}{enumerate}{1}
\setlist[SDom]{label=\textbf{SD},ref={\text{SD}}}

\begin{SDom}[align=parleft, leftmargin=!,itemsep=0pt,labelsep=14pt] 
        \item\label{SDom}  There is some $(x,y)\in \Graph$ that weakly dominates every $(s,t) \in \Graph$.  
\end{SDom}

Both of these axioms are stronger than \ref{CDom}:
\[
	\ref{Ar} \implies \ref{SDom} \implies \ref{CDom}.
\]
The first implication is trivial when $\sGraph$ is nonempty. When it is empty, 
both \ref{Ar} and \ref{SDom} hold automatically, in the latter case because
every $(x,x)$ in $\Graph$ weakly dominates every $(s,t)$ in $\Graph$. 
For the second implication, notice that \ref{SDom} is the special case of \ref{CDom} when $D$ is a singleton. 
The implications, however, are irreversible, as 
shown by the next example. Further examples contrasting \ref{Ar}, \ref{SDom} and \ref{CDom} will be given below.

\begin{example}\label{ex:CDnotSD}
Let $S$ be a non-empty set, and 
$M=\RRR^S_0$, the vector space of finitely-supported functions $S\to\RRR$. As a vector space, it is 
also a mixture space. Define a  mixture  preorder on $M$ by 
\[f\se g \iff f(s)\geq g(s)\text{ for all }s\in S.\]
Then $\se$ satisfies \ref{MC}. It satisfies \ref{Ar} if and only if $|S|= 1$. It satisfies \ref{SDom} if and only if $|S|$ is finite; it satisfies \ref{CDom} if and only if $|S|$ is countable. 
To illustrate when $|S|$ is countable, 
define  $D=\Set{(1_A,0)}{A\subset S, A\text{ finite}}$, where $1_A\in M$ is the characteristic function of $A$. 
Then $D$ is a countable subset of $\Graph$, and each $(f,g)\in \Graph$ is weakly dominated by 
the element 
$(1_{\supp (f-g)},0)$ 
of $D$. 
\end{example}

The axioms \ref{Ar}, \ref{SDom}, and \ref{CDom} can also be reformulated in terms of `Archimedean classes', an idea usually developed in the context of ordered groups or vector spaces  \citep[see e.g.][]{HW1952}. 
In the present context of preordered mixture spaces, let us say two pairs 
$(x,y)$ and $(s,t)$ in $\Graph$ 
are in the same \emph{Archimedean class} if each weakly dominates the other (this is an equivalence relation, since weak domination is a preorder). 
Write $[(x,y)]$ for the Archimedean class of $(x,y)$, and let 
$\ArchCl$ be the set of Archimedean classes 
in $\Graph$. What we call the {\em Archimedean structure} of a mixture preorder $\se$ is the partially ordered set $(\ArchCl, \geq)$ where 
$[(x,y)] \geq [(s,t)]$ if and only if $(s,t)$ weakly dominates $(x,y)$.%
\footnote{The direction of the inequality may be surprising, but it is standard in the related literature on valuation theory, and may be thought of as saying that 
$(s,t)$ comes earlier in order of importance than $(x,y)$.
}
Note that $\ArchCl$ always contains a maximal element, the single Archimedean class consisting of all pairs $(x,y)$ with $x\e y$.
As the following easily proved equivalences show,  
\ref{Ar}, \ref{SDom}, and \ref{CDom} can all be seen as placing limits on the 
complexity of the Archimedean structure.  
\begin{enumerate}
[label=(\alph*)]
\item\label{p:val:Ar} $\se$ satisfies \ref{Ar} if and only if $(\ArchCl, \geq)$ has at most two elements. 
\item\label{p:val:SD} $\se$ satisfies \ref{SDom} if and only if  $(\ArchCl, \geq)$  contains a minimum element. 
\item\label{p:val:CD} $\se$ satisfies \ref{CDom} if and only if $(\ArchCl, \geq)$ contains a countable coinitial subset.%
\footnote{Recall that a subset $S'$ of a preordered set $(S, \se_S)$ is {\em coinitial} if and only if, for every $s\in S$, there exists $s'\in S'$ with 
$s \se_S s'$.}
\end{enumerate}
Specifically, if \ref{CDom} holds with respect to a countable  $D\subset \Graph$, then $\Set{[(x,y)]}{(x,y)\in D}$ is a countable coinitial subset of $\ArchCl$.

There are of course many other ways of limiting the complexity of Archimedean structures, but these are the ones of immediate interest.  Appendix~\ref{s:weak dom} provides more formal discussion of Archimedean 
structures; here we illustrate with some examples.

\begin{example}\label{ex:arch} In Example~\ref{ex:CDnotSD}, 
for $(f,g)$, $(h,k) \in \Graph$, 
$(f,g)$ weakly dominates $(h,k)$  if and only if  $\supp(f-g)\supset\supp(h-k)$.
Therefore $[(f,g)]\mapsto \supp(f-g)$ is an isomorphism between the Archimedean structure $(\ArchCl,\geq)$ and the set of finite subsets of $S$, partially ordered by $\subset$. 
The results of appendix~\ref{s:weak dom} yield a different description. Consider the convex cone of positive functions,%
\footnote{Convex cones are defined and discussed in section~\ref{s:mspace}.} 
$C = \Set{f \in \RRR^S_0}{f \geq 0}$. For each finite $A\subset S$, $F_A=\Set{f\in C}{\supp(f)\subset  A}$ is a face of $C$.
If $f\in C$,  then $F_{\supp(f)}$ is the smallest face containing $f$;  this shows that the faces of the form $F_A$, with $A$ finite, are what we call the \emph{regular} faces of $C$. Clearly $A\subset B\iff F_A\subset F_B$. So we conclude that $(\ArchCl,\geq)$ is isomorphic to the set of regular faces of $C$, partially ordered by $\subset$.   Proposition \ref{p:faces} generalizes this description.  
It also notes that $(\ArchCl, \geq)$ has at most one minimal element, corresponding to the largest (thus $\subset$-minimal) face $C$, if it is regular. 
In the present example, 
it has a minimal element only if 
$S$ is finite.  
\end{example}

The following example of a lexicographically ordered vector space makes the structure of 
$(\ArchCl, \geq)$ 
particularly clear (but \ref{MC} is not generally satisfied): 
\begin{example}\label{ex:LFS}
Let $(S,\geq)$ be an  
ordered set, and 
as in Example~\ref{ex:CDnotSD}, let $M=\RRR^S_0$ 
be the set of finitely supported functions $S \to \RRR$. 
For distinct $f$ and $g$ in $M$, 
let $s(f,g)=\min\Set{s\in S}{f(s)\neq g(s)}$.   
Define a mixture preorder on $M$ by
\[f\se g\iff \text{either }f=g,\text{ or }f(s(f,g))\geq g(s(f,g)).\]
Let $\ArchCls\subset\ArchCl$ be the set of Archimedean classes of \emph{strictly} positive pairs, i.e.~the $[(f,g)]$ with $f\s g$. It merely omits the maximal element of $\ArchCl$. 
One can then see that $[(f,g)]\mapsto s(f,g)$ is an isomorphism of ordered sets between $\ArchCls$ and $S$. Thus \ref{Ar} holds if and only if $|S|\leq 1$; \ref{SDom} holds if and only if $S$ contains a minimal element, e.g. if $S=\NNN$; and \ref{CDom} holds if and only if $S$ contains a countable coinitial subset, e.g.~if $S=\RRR$.
\end{example}

\begin{rem}
Most of our examples in this section concern vector spaces. However, this is only for simplicity. Indeed, if $(M,\se)$ is a preordered mixture space (a vector space or otherwise), and $M'$ is any mixture space of the same dimension, then there is a mixture preorder on $M'$ with the same Archimedean structure as $\se$, and which satisfies \ref{MC} or \ref{MR} if and only if $\se$ does. (This follows from Propositions  \ref{p:faces}\ref{p:faces:iso}  and \ref{p:cone} below.)
\end{rem}

\subsection{Countable domination and countable dimension}\label{s:count}
As already mentioned, \ref{CDom} strictly weakens the requirement that the dimension of $M$ be countable; 
we prove the following in appendix \ref{s:aux proofs}:

\begin{prop} \label{p:dim}
If a preordered mixture space has countable dimension, then it satisfies \ref{CDom}.  The converse does not hold, even for mixture preorders that satisfy \ref{MC}. 
\end{prop}

We first illustrate why the converse of Proposition~\ref{p:dim} fails, in particular 
for mixture preorders that satisfy \ref{MC}. 
One reason is that the dimension of a mixture space can always be increased by introducing extra dimensions of indifference or incomparability,  as the following example shows.

\begin{example}\label{ex:combo}
Let $(M_1,\se_1)$, $(M_2,\se_2)$, $(M_3,\se_3)$ be preordered mixture spaces.  Assume that $\se_2$ is complete indifference ($x\e_2 y$ for all $x,y\in M_2$), and $\se_3$ is complete incomparability ($x\se_3 y$ only 
if $x=y$ for $x,y \in M_3$).
Note that $\se_2$ and $\se_3$ both satisfy \ref{MC}. 
Define a preordered mixture space $(M,\se)$ by letting $M$ be the product $M = M_1\times M_2\times M_3$, 
with the mixture operation defined component-wise, 
and $\se$ be the product preorder. Thus in this case
\[
(x_1,x_2,x_3)\se (y_1,y_2,y_3)\iff x_1\se_1 y_1 \text{ and }x_3=y_3.
\]
It is easy to check that $\se$ satisfies \ref{MR}, \ref{MC}, or \ref{CDom}  if and only if $\se_1$ does, and that $\dim M = \dim M_1+\dim M_2+\dim M_3$. 
Suppose that $\se_1$ satisfies \ref{MC} and that $M_1$ has countable dimension. 
By Theorem~\ref{t:count} and Proposition~\ref{p:dim}, $\se_1$ will satisfy \ref{MR} and \ref{CDom}. 
Thus $\se$ will also satisfy \ref{MC}, \ref{MR} and \ref{CDom}, but $M$ may have arbitrarily high dimension.  
\end{example}

However, the following example shows that we can have 
\ref{MC} and \ref{CDom} (and hence \ref{MR}), 
and arbitrarily high dimension, even if there is no decomposition of the type just illustrated. 
\begin{example}\label{ex:SDnotAr}
Let $W$ be a nontrivial 
normed vector space, and $M=W\times \RRR$: as a vector space, $M$ is also a mixture space. 
Define a mixture preorder on $M$ by 
\[(v,a)\se (w,b) \iff |v-w|\leq a-b.\]
Then $\se$ satisfies 
\ref{MC}, \ref{SDom}, and hence \ref{CDom}, but not \ref{Ar}. We can take $D=\{(0,1;0,0)\}$. 
In this case, however, $\dim M=\dim W+1$, which can be arbitrarily large. There is no nontrivial indifference ($x\e y\implies x=y$). Although there is incomparability,  note that, for any $x,y\in M$, there is some $z\in M$ with $z\se x$ and $z\se y$. This would not be true if $M$ were a product of preordered mixture spaces with a nontrivial, completely incomparable factor. 

For a similar example in which \ref{MC} and \ref{CDom} hold, but \ref{SDom} (and hence \ref{Ar}) does not, take the $M$ just described and let $M'=M^\NNN_0$, the set of finitely supported functions $\NNN \to M$. Define a mixture preorder $\se'$ on $M'$ by $f \se' g$ if and only if $f(n) \se g(n)$ for all $n$.
In analogy to Example~\ref{ex:CDnotSD} (in the case of $S$ countably infinite), \ref{CDom} holds for $\se'$ with respect to  $D= \Set{(0,1_A;0,0)}{A \subset \NNN, A \text{ finite}}$.   
\end{example}


Turning to ways in which Proposition~\ref{p:dim} may be strengthened,  Example~\ref{ex:combo} may suggest the conjecture 
that \ref{CDom} holds if, for every $x\in M$, the mixture sets 
$\Set{y\in M}{y\se x}$ and
$\Set{y\in M}{x\se y}$ have countable dimension. However, Example \ref{ex:klee2} will provide a counterexample to this conjecture; in it, those sets even have finite dimension. 
Nevertheless, as we explain in Remark~\ref{rem:CD}, there is a precise sense in which \ref{CDom} is a dimensional restriction.

\section{Proofs of main results and discussion}\label{s:discuss}

\subsection{From mixture spaces to vector spaces}\label{s:mspace}
Our motivation for studying mixture spaces was given in the introduction. 
However, at a technical level, we will use a standard method to reduce questions about mixture spaces to equivalent, but mathematically more convenient, questions about 
vector spaces. In the context of multi-representations, this reduction was first used in \citet{SB1998}.\footnote{
Besides \citet{SB1998}, we refer the reader to \citet{pM2001} for a careful study of the embedding it relies on, and to a text such as \citet{eO2007} for the vectorial concepts.}

It follows from a standard embedding theorem%
\footnote{See \citet[\S 3]{mH1954}. A more general embedding theorem was given in \citet[Thm. 2]{mS1949}, but Hausner's result is easier to apply directly.} 
that any mixture space $M$ can be embedded in a (real) vector space $V$, in such a way that $V$ is the affine hull of $M$ (so $V=\Span(M-M)$), and the mixture operation on $M$ coincides with that on $V$: $x\aa y= \aa x+(1-\aa)y$. $M$ is, therefore, a convex subset of $V$,
and from this it is easy to show
\begin{equation}\label{eq:Vspan}
	V=\Set{\lambda(x-y)}{\lambda>0, \,x,y\in M}.
\end{equation}
We follow Shapley and Baucells in calling such an embedding {\em efficient}. Efficient embeddings are essentially unique: if $M\subset V$ and $M\subset V'$ are efficient embeddings, then there is a unique affine isomorphism $V\to V'$ that is the identity map on $M$.

Recall that a {\em linear preorder} $\seV$ on a vector space $V$ is a preorder on $V$ that is compatible with vector addition and positive scalar multiplication;  that is, $v \seV v' \iff \ll v + w \seV \ll v' +w$ for all $v,v',w\in V$ and $\ll >0$. Let $M \subset V$ be an efficient embedding. (Considering $V$ as a mixture space, a linear preorder is the same as a mixture preorder.)
As Shapley and Baucells explain,
  there are natural one-to-one correspondences between mixture preorders $\se$ on $M$, convex cones
$C\subset V$, and linear preorders $\seV$ on $V$, such that, for all $x,y\in M$,%
\footnote{Since terminology varies slightly: $C\subset V$ is a convex cone if and only if $C$ is 
nonempty, 
convex and $[0,\infty)C=C$. 
\label{n:SB1998}We note that although Shapley and Baucells start with axioms that are different from ours (see appendix~\ref{s:WCon}), they first derive \ref{SI} from their axioms, then use \ref{SI} to construct the correspondences we describe here.
The correspondence between $\se$ and $C$ is  stated in
their equations (11) and (12);  the well-known correspondence between $C$ and $\seV$ follows if we consider $V$ as a mixture space. 
} 
\begin{equation}\label{eq:mspace}
x \se y \iff x-y \in C \iff x \seV y.
\end{equation}	
This formula explicitly defines $\se$ in terms of $\seV$ or $C$, while the next formulae explicitly define $C$ in terms of $\se$, and $\seV$ in terms of $C$:
\begin{equation}\label{eq:CseVdef}
	C=\Set{\lambda(x-y)}{\lambda>0, \,x\se y} \qquad v \seV 0 \iff v \in C.
\end{equation}
We then call $C$ the {\em positive cone} of $\se$, 
and $\seV$ the {\em linear extension} of $\se$. 

Finally, mixture-preserving functions $u\colon M\to \RRR$
correspond one-to-one with affine functions $\tilde u\colon V\to \RRR$, in such a way that 
$\tilde u$ extends $u$, that is, 
$\tilde u\restr{M}=u$. 
Moreover, a set $\UU$ of mixture-preserving functions $M\to\RRR$ is a multi-representation of $\se$ if and only if  $\Set{\tilde u}{u\in \UU}$ is a mixture-preserving multi-representation of $\seV$.  It follows from \eqref{eq:CseVdef} that an equivalent condition in terms of $C$ is 

\begin{equation}\label{eq:intersection}
	C=\bigcap_{u\in\UU} \Set{v\in V}{\tilde u(v)\geq \tilde u(0)}.
\end{equation}



\subsection{Proofs of main results}\label{s:main proofs}
We now prove our main results in terms of a series of auxiliary results. We outline the ideas on which the auxiliary results are based, but unless otherwise stated, we defer their full proofs to appendix \ref{s:aux proofs}. Given the existence of efficient embeddings, 
our positive results mainly rely on standard extension and separation techniques in vector spaces. The proofs of the negative results are more striking, and we describe the counterexamples on which they are based.

\subsubsection{Preliminaries}\label{s:prelim}
Recall that a subset $S$ of a vector space  $V$ is {\em algebraically closed} 
 if $v \in S$ whenever 
  $(v,w]\subset S$. (In standard notation, $(v,w] =\Set{(1-\aa)v+\aa w}{ \aa\in (0,1]}$.) 
We say that $S\subset V$ is {\em weakly closed} in $V$ if it is closed in the weak topology on $V$ 
(see Remark \ref{rem:weak}). 
We prove the following proposition in appendix~\ref{s:aux proofs}.

\begin{prop}\label{p:cone}
\hspace{0pt}
Let $(M,\se)$ be a preordered mixture space, $M\subset V$ an efficient embedding, and $C\subset V$ the positive cone. 
 \begin{enumerate}
[nosep, label=(\roman*)]
\item\label{p:cone:dim} $\dim M$ equals the vector-space dimension of $V$.
\item\label{p:cone:MC} $\se$ satisfies \ref{MC} if and only if $C$ is algebraically closed.
\item\label{p:cone:MR}  $\se$ satisfies \ref{MR} if and only if $C$ is weakly closed in $V$.  
\end{enumerate}
\end{prop}
\noindent Part~\ref{p:cone:dim} shows that our definition of the dimension of $M$ in section~\ref{s:main} is equivalent to a more standard characterisation \citep[see e.g.][]{pM2001}. 
Part~\ref{p:cone:MC} is  almost the same as \citet[Thm. 1.6]{SB1998}, but since our axioms are slightly different, we provide a proof. In fact, we will use \ref{p:cone:MC} to show that our axioms are equivalent to theirs, in Appendix \ref{s:WCon}. Part~\ref{p:cone:MR} is a routine application of the strong separating hyperplane theorem.

\subsubsection{Theorems~\ref{t:MC not suff} and~\ref{t:klee}}

\begin{proof}[{\bf Proof of Theorem~\ref{t:MC not suff}}]
The proof that \ref{MR} implies \ref{MC} is straightforward. Indeed, suppose that $\se$ has a mixture-preserving multi-representation $\UU$. Suppose given $x,y,z\in M$ such that $x\aa y\s z$ for all $\aa\in(0,1]$. Then, for any $u\in\UU$, $u(x\aa y) \geq u(z)$.  But $u(x\aa y)
=\aa u(x)+(1-\aa)u(y)$. 
In the limit $\aa\to 0$, we find $u(y)\geq u(z)$. Since this is true for all $u\in\UU$, we must have $y\se z$, as required for \ref{MC}.

The fact that the converse fails is immediate from Theorem~\ref{t:klee}, to which we now turn.
\end{proof}

The proof of Theorem \ref{t:klee} appeals to the following proposition, further discussed below.   

\begin{prop}\label{p:klee}
Let $V$ be a vector space of uncountable dimension. There exists a convex cone in $V$ that is algebraically closed but not weakly closed in $V$. 
\end{prop}

\begin{proof}[{\bf Proof of Theorem~\ref{t:klee}}]
Let $M \subset V$ be an efficient embedding of a mixture space $M$ of uncountable dimension, so that, by Proposition~\ref{p:cone}\ref{p:cone:dim}, $V$ also has uncountable dimension. By Proposition~\ref{p:klee}, $V$ contains a convex cone that is algebraically closed but not weakly closed. Using \eqref{eq:mspace}, this cone defines a mixture preorder on $M$. By Proposition~\ref{p:cone} parts \ref{p:cone:MC} and \ref{p:cone:MR}, this mixture preorder satisfies \ref{MC} but not \ref{MR}.
\end{proof}

We prove Proposition~\ref{p:klee} in appendix \ref{s:aux proofs}. The proof 
relies on following example, 
which is based on \citet{vK1953}. 
Klee showed that if a vector space has uncountable dimension, then it contains a convex subset that is algebraically closed but not weakly closed 
(see \citet[pp. 194--195]{gK1969} for a discussion in more modern terminology). We modify Klee's construction
to obtain a convex cone with similar properties.  

\begin{example}\label{ex:klee}
Let $V$ be a vector space with an uncountable basis $B$. 
Endow $V$ with the weak topology. 
Given a subset $S$ of $V$, we write $\cone(S)$ for the convex cone in $V$ generated by $S$, that is, the smallest convex cone that contains $S$. 
It consists of all linear combinations of $S$ with non-negative coefficients. 
Choose $b_0\in B$, and let $B_1=B\setminus\{b_0\}$. For each finite, non-empty subset $A\subset B_1$, let $y_A = |A|^{-2}\sum_{b\in A} b$.   
Define a convex cone 
\[
	K = \cone{\Set{y_A + b_0}{A\subset B_1 \text{ is nonempty and finite}}}.
\]
The proof of Proposition~\ref{p:klee} shows that $K$ is algebraically closed but not 
closed. In fact, this generalizes slightly: 
the same argument, using separating hyperplanes, shows that $K$ is not closed with respect to \emph{any} locally convex topology on $V$. 
\end{example}

\subsubsection{Theorem~\ref{t:count}}
The proof rests on the following, which  
provides a converse to the result of Klee just mentioned.

\begin{prop}\label{p:count}
Let $V$ be a vector space of countable dimension.  Every convex  set 
in $V$ that is algebraically closed is weakly closed in $V$. 
\end{prop}
This was proved using the algebraic version of the separating hyperplane theorem in  \citet[(3) on p. 194]{gK1969}. In appendix~\ref{s:aux proofs} 
we provide a slightly different proof: to apply the separating hyperplane theorem, we use a result of \citet{vK1953}, that in a vector space of countable dimension, the finite topology is locally convex.

\begin{proof}[{\bf Proof of Theorem~\ref{t:count}}] Suppose that \ref{MC} holds and that $M$ has countable dimension. Given an efficient embedding $M\subset V$, $V$ also has countable dimension, by Proposition~\ref{p:cone}\ref{p:cone:dim}. By Proposition \ref{p:cone}\ref{p:cone:MC}, the positive cone $C$ is algebraically closed, so, by Proposition \ref{p:count}, it is weakly closed. Therefore, by Proposition \ref{p:cone}\ref{p:cone:MR}, $\se$ satisfies \ref{MR}. 

For a counterexample to the converse implication, let $M$ be an uncountable-dimensional vector space with the preorder of complete indifference: $x\e y$ for all $x,y\in M$. This satisfies \ref{MR} despite having uncountable dimension. (Examples \ref{ex:combo}, \ref{ex:SDnotAr} 
and \ref{ex:cantor} 
provide less trivial examples.)
\end{proof}


\subsubsection{Theorem~\ref{t:CD}}\label{s:t:CD}
We first interpret \ref{CDom} and, for future reference, \ref{SDom}, in terms of the positive cone.  
For further discussion of Archimedean structure along similar lines, see appendix~\ref{s:weak dom}. 
 Let $V$ be a vector space with linear preorder $\seV$; let $C$ be any subset of $V$.  Recall that the {\em relative algebraic interior} of $C$ consists of those $v\in C$ with the following property: 
for every $w \in \aff(C)$, 
the affine hull of $C$, 
there is some $\ee >0$ such that $[v, v+\ee w] \subset S$. 
 
Recall also that a set $S$ is \emph{cofinal in $C$} (with respect to $\seV$) if $S\subset C$ and, for all $v\in C$, there is some $s\in S$ with $s \seV v$. 



\begin{prop}\label{p:misc}
Let $(M\se)$ be a preordered mixture space, $M\subset V$ an efficient embedding, $C$ the positive cone, and $\seV$ the  linear extension. 
\begin{enumerate}
[nosep, label =(\roman*)]
\item\label{p:misc:SD} $\se$ satisfies \ref{SDom} if and only if $C$ has a nonempty relative algebraic interior.
\item\label{p:misc:CD} $\se$ satisfies \ref{CDom} if and only if there is a countable set that is cofinal in $C$.
\end{enumerate}
\end{prop}

We will also use the following standard extension theorem, due to \citet{lK1937}. For a proof, see \citet[Thm. 1.36]{AT2007}. 
\begin{thm}[Kantorovich]\label{t:kantorovich}
Let $V$ be a vector space with a linear preorder $\seV$.  
Let $W$ be a linear subspace that is  cofinal in $V$. 
Then any increasing 
linear functional on $W$ extends to 
an increasing linear functional on $V$. 
\end{thm}

\begin{proof}[{\bf Proof of Theorem \ref{t:CD}}]
We first give a counter-example to the reverse implication;  that is, we give an example 
of a mixture preorder that satisfies \ref{MR} (and therefore \ref{MC}) but not \ref{CDom}.  

\begin{example}\label{ex:cantor} Let $M=\RRR^\NNN$, the vector space of functions $\NNN \to \RRR$. Define a mixture preorder on $M$ by $f \se g \Leftrightarrow f(n) \geq g(n)$ for all $n \in \NNN$. This clearly satisfies \ref{MR} (the canonical projections $\RRR^\NNN\to \RRR$ provide a multi-representation), but it violates \ref{CDom}. 
{\em Proof:} In this case, the positive cone $C$ consists of the $f\in M$ with $f(n)\geq 0$ for all $n$. 
Suppose that \ref{CDom} holds; by  Proposition~\ref{p:misc}\ref{p:misc:CD}, there is a 
countable subset $\{f_1,f_2,\ldots\}$ cofinal in $C$. 
Let $f(k)=f_k(k)+1\in C$. Then for no $k$ is it true that $f_k\se f$, a contradiction. 
 \end{example}

Now let $(M,\se)$ be a preordered mixture space, satisfying \ref{MC} and \ref{CDom}; we have to show it satisfies \ref{MR}. Let $M \subset V$ be an efficient embedding, $C$ the positive cone, and $\seV$ the linear extension. For any subspace $W\subset V$ we let $\se_W$ be the restriction of $\seV$ to $W$, a linear preorder with positive cone $C_W=C\cap W$. 

By Proposition~\ref{p:misc}\ref{p:misc:CD}, 
there is a countable 
set 
$Z$ cofinal in $C$.
Given $w \in V\setminus C$, set $Z_w = \Span(Z \cup \{w\})$. 
By Proposition~\ref{p:cone}\ref{p:cone:MC}, $C$ is algebraically closed.  It follows that $C_{Z_w}$ is also algebraically closed. 
Since $Z_w$ has countable dimension, 
$C_{Z_w}$ is weakly closed in $Z_w$, 
by Proposition~\ref{p:count}. 
By the strong separating hyperplane theorem 
\citep[Cor. 5.84]{AB2006}, 
there is a linear functional $L'_w$ on $Z_w$ such that $L'_w(C_{Z_w}) \subset [0, \infty)$  and $L'_w(w) < 0$. 
Because $L'_w(C_{Z_w}) \subset [0, \infty)$,
$L'_w$ is  an increasing linear functional on $Z_w$. 

Let  $Y_w =\Span(C \cup \{w\})$. We claim that $Z_w$ 
is cofinal in $Y_w$. 
Indeed, let $y\in Y_w$. We can write it in the form $y=\ll w+ \sum_{c\in C} \ll_c c$, with $\ll,\ll_c\in\RRR$ and finitely many  $\ll_c$ being non-zero. 
Since $Z$ is cofinal in $C$,
we can find, for each $c\in C$, some $z_c\in Z$ with $z_c\seV c$. Since $c\seV 0$, it follows that $|\ll_c|z_c\seV \ll_c c$. Therefore $\ll w+\sum_{c\in C}|\ll_c| z_c\seV y$. Since the left-hand side of this formula is an element of $Z_w$, $Z_w$ is cofinal in $Y_w$. 

By Theorem~\ref{t:kantorovich}, $L'_w$ extends from $Z_w$ to 
an increasing linear functional $L''_w$ on $Y_w$. Arbitrarily extend $L''_w$ to a linear functional $L_w$ on $V$. 
By construction, $L_w(C)\subset [0,\infty)$ and $L_w(w)<0$. Therefore 
$C=\bigcap_{w\in{V\setminus C}}\Set{v\in V}{L_w(v)\geq 0}$. It follows  from \eqref{eq:intersection} that  $\UU = \Set{L_w\restr{M}}{w \in V\setminus C}$ is a mixture-preserving multi-representation of $\se$. 
\end{proof}

\begin{rem}\label{rem:CD}
The following variation on Proposition \ref{p:misc}\ref{p:misc:CD}, 
also proved in appendix~\ref{s:aux proofs}, 
explains the sense in which \ref{CDom} is a dimensional restriction, generalizing the countable dimensionality condition used in Theorem \ref{t:count}. 
\begin{cor}\label{c:CD dim}
Let $(M\se)$ be a preordered mixture space, $M\subset V$ an efficient embedding, $C$ the positive cone, and $\seV$ the  linear extension. 
Then $\se$ satisfies \ref{CDom} if and only if there is a subspace that is cofinal in $\Span C$ and that has countable dimension. 
\end{cor}
\noindent
To illustrate: in Example~\ref{ex:SDnotAr}, $\Span(C) = M$, which may have arbitrarily high dimension, but $\Span \{(0,1)\}$ is a one-dimensional cofinal subspace.  
\end{rem}

\subsubsection{Theorem~\ref{t:top}}\label{s:top}
We begin with a mostly well-known observation that 
generalizes some of the claims in 
Theorem~\ref{t:top}.  
Say that a preorder $\se$ on an arbitrary topological space $M$ has a {\em continuous multi-representation} if it satisfies

\newlist{CMR}{enumerate}{1}
\setlist[CMR]{label=\textbf{CMR},ref={\text{CMR}}}

\begin{CMR}[align=parleft, leftmargin=!,itemsep=0pt,labelsep=14pt] 
\item\label{CMR} There is a nonempty 
set $\UU$ of continuous 
functions $M\to\RRR$, such that for all $x,y\in M$, 
\[
	x \se y \iff u(x) \geq u(y) \text{ for all } u \in \UU.
\]
\end{CMR}

\begin{lemma}\label{l:top} Let $\se$ be a preorder on a topological space $M$.  Then
$
\ref{CMR} \implies \ref{Cl} \implies \ref{Con}. 
$
Moreover, suppose $M$ is a mixture space such that, for each $x, y \in M$, the map $f_{x,y}\colon [0,1]\to M$ given by $\aa\mapsto x\aa y$ is continuous. Then $\ref{Con} \implies \ref{MC}.$
\end{lemma}

The proof of Lemma \ref{l:top} is in appendix \ref{s:aux proofs}. Here we use it to deduce Theorem \ref{t:top}.

\begin{proof}[{\bf Proof of Theorem~\ref{t:top}}]
If $M$ is a mixture space with the weak topology, then every mixture-preserving function $M\to\RRR$ is continuous; therefore \ref{MR} implies \ref{CMR}. Moreover, for each $x,y\in M$, the map $f_{x,y}\colon [0,1]\to M$ given by $\aa\mapsto x\aa y$ is continuous.  The implications stated in Theorem~\ref{t:top} are therefore immediate from Lemma~\ref{l:top}. 

To show that the third implication in Theorem \ref{t:top} cannot be reversed, we need an example that satisfies \ref{MC} but not \ref{Con}. 
We again appeal to Example \ref{ex:klee}. We take $M=V$ and let $\se$ be the mixture preorder with positive cone $C=K$. 
Recall that  $K$ is algebraically closed but not weakly closed
(as shown in proving Proposition~\ref{p:klee}).  
By Proposition \ref{p:cone}\ref{p:cone:MC}, $\se$ satisfies \ref{MC}. 
Since $K = \Set{x \in M}{x \se 0}$, $\se$ violates \ref{Con}.

Finally, we need to show that \ref{Con} does not imply \ref{Cl}. We isolate this claim as the following proposition and prove it separately.  
\end{proof}

\begin{prop}\label{p:klee2}
There is preordered mixture space $(M, \se)$ such that $\se$ is continuous but not closed in the weak topology on $M$.
\end{prop}
\noindent
The proof of Proposition~\ref{p:klee2}, given in appendix \ref{s:aux proofs}, involves the following modification of Example~\ref{ex:klee}. 


\begin{example}\label{ex:klee2}
Let $V$, $B$, and $K$ be as in Example~\ref{ex:klee}. 
Let $V^+ = \cone (B)$. 
For any $v\in V^+$, let $A_v\subset B_1$ be the set of elements of $B_1$ with respect to which $v$ has strictly positive coefficients. 
Let $V_v=\Span(A_v\cup\{b_0\})$, and 
\[
M=\Set{(v,w)}{v\in V^+, w\in V_v}\subset V\times V.
\]
This $M$, it is easy to check, is a convex set. 
Let $\se$ be the mixture preorder on $M$ with the positive cone
$K'=\Set{(0,w)}{w\in K}\subset V\times V$. 
That is, for all  $(x,y), (v,w) \in M \times M$,
\begin{equation}\label{eq:klee2}
(x,y) \se (v,w) \iff x-v=0,\, y-w\in K\cap V_v.
\end{equation} 
Equip $M$ with the weak topology. 
The proof of Proposition~\ref{p:klee2} consists in the verification that $\se$ is continuous but not closed. 
\end{example}

\begin{rem}\label{rem:klee} 
 Let $(M, \se)$ be a preordered mixture space with the weak topology. 
As already noted, by Theorems~\ref{t:CD} and~\ref{t:top}, the conditions \ref{MR}, \ref{Cl}, \ref{Con} and \ref{MC} are equivalent when \ref{CDom} holds. In addition, when $M$ is a vector space, the conditions  \ref{MR}, \ref{Cl}, and \ref{Con} (but not \ref{MC}) are equivalent. 
To show the equivalence, it is sufficient, by Theorem~\ref{t:top},  to show that \ref{Con} entails \ref{MR}. 
Since $M$ is a vector space, $\se$ is a linear preorder, with corresponding positive cone $C = \Set{x \in M}{x \se 0}$. But \ref{Con} implies that $C$ is closed, implying \ref{MR} by Proposition~\ref{p:cone}\ref{p:cone:MR}. 
\end{rem}

\section{Strict multi-representation and uniqueness}\label{s:refine}
We now briefly discuss two standard topics concerning mixture-preserving multi-representations.

\subsection{Strict multi-representation}\label{s:strict}

The pioneering study of expected utility without the completeness axiom of \citet{rA1962} focussed on the existence of a single real-valued, strictly increasing, mixture-preserving function  (as defined in section~\ref{s:main});  see also \citet{pF1982}. But such  a function does not  fully characterize an incomplete preorder,  and interest turned to the existence of mixture-preserving multi-representations,  which do.
One can try to combine these approaches by considering mixture-preserving multi-representations that consist entirely of strictly increasing functions: 
\newlist{SMR}{enumerate}{1}
\setlist[SMR]{label=\textbf{SMR},ref={\text{SMR}}}
\begin{SMR}[align=parleft, leftmargin=!,itemsep=0pt,labelsep=14pt] 
\item\label{SMR} There is a nonempty 
set $\UU$ of \emph{strictly increasing}
mixture-preserving functions $M\to\RRR$, such that for all $x,y\in M$, 
\[
	x \se y \iff u(x) \geq u(y) \text{ for all } u \in \UU.
\]
\end{SMR}
The advantages of such `strict' multi-representations have been emphasized by \citet{oE2014} and  \citet{lG2017}, although Evren uses a notion of multi-representation that is 
different from ours. 
We now present two basic results about extending \ref{MR} to \ref{SMR}. Since our earlier results gave sufficient conditions for \ref{MR}, results giving sufficient conditions for \ref{SMR} are implied. 

 
First, we note that if a mixture preorder satisfies \ref{MR}, then solving Aumann's problem---that is, finding a strictly increasing mixture-preserving function---is enough to guarantee \ref{SMR} as well.   

\begin{prop}\label{p:strict1}
Let $(M,\se)$ be a preordered mixture space. Then $\se$ satisfies \ref{SMR} if and only if  it satisfies \ref{MR} and there exists a strictly increasing mixture-preserving function $M\to\RRR$.
\end{prop}

The second result extends the 
picture given by 
Theorems \ref{t:klee} and \ref{t:count} to representations by strictly increasing functions.

\begin{prop}\label{p:strict2} Let $M$ be a mixture space.  
\begin{enumerate}[nosep, label= (\roman*)]
\item\label{p:strict2:count-a} If $\dim M$ is countable, any mixture preorder on $M$ that satisfies \ref{MR} also satisfies \ref{SMR}.    
\item\label{p:strict2:uncount-a}
If $\dim M$ is uncountable, there is a mixture preorder on $M$ that satisfies \ref{MR} but not \ref{SMR}.  
\end{enumerate}
\end{prop}

In common with our earlier results,  these results show a sharp difference between the cases of countable and uncountable dimension. Theorem~\ref{t:count} and Proposition~\ref{p:strict2} together show that, when $\dim M$ is countable, \ref{MC} is equivalent to \ref{SMR}. But when $\dim M$ is uncountable, \ref{MC} is not sufficient even for \ref{MR}; and even if \ref{MR} is satisfied, \ref{SMR} may not be.


The proof of Proposition~\ref{p:strict1} is very simple. For Proposition~\ref{p:strict2}, the main idea of the proof 
of \ref{p:strict2:count-a}
is that countable dimension enables us to focus on multi-representations with countably many elements, as the following lemma shows. Such a countable multi-representation can be used to construct a strictly increasing function, and Proposition \ref{p:strict1} applies.


\begin{lemma}\label{l:card}
Let $(M,\se)$ be a preordered mixture space. If $\se$ has a mixture-preserving multi-representation $\UU$, then it has a mixture-preserving multi-representation $\UU' \subset \UU$ such that $| \UU'| \leq \max(\aleph_0, \dim M)$.
\end{lemma}

The proof 
of Proposition \ref{p:strict2}\ref{p:strict2:uncount-a}
rests on the following example. 

\begin{example}\label{ex:no strict}
Assume that $\dim M$ is uncountable. Let $M\subset V$ be an efficient embedding, so $\dim V$ is uncountable. 
For some uncountable ordinal $\kappa$, we can choose a basis $\Set{v_\aa}{\aa < \kappa} \subset M$ for $V$ indexed by ordinals $\aa$ smaller than $\kappa$. 
For each $\bb<\kappa$, let $\pi_\bb$ be the unique linear functional on $V$ such that $\pi_\bb(v_\aa) = 
1$ if $\aa=\bb$ and $\pi_\bb(v_\aa)=0$ otherwise. 
For each $\aa < \kappa$, define a 
mixture-preserving function $u_\aa$ on $M$ by $u_\aa(x) = \sum_{\bb \leq \aa}\pi_\bb(x)$. 
This is well-defined, since for each $x$ in $V$, and hence $M$, 
$\pi_\bb(x)$ is nonzero for only finitely many $\bb$. 
Let $\UU = \Set{u_\aa}{\aa<\kappa}$, and let $\se$ be the mixture preorder on $M$ that it represents. 
The proof of Proposition~\ref{p:strict2}\ref{p:strict2:uncount-a} shows that $\se$ does not have a strictly increasing  function, mixture-preserving or otherwise. 
\end{example}



\subsubsection{Related literature}\label{s:strict related lit}

Suppose that $M$ is a preordered mixture space of uncountable dimension. \citet{rA1962} showed that the continuity condition \ref{Au} (see  note~\ref{n:aumann}), which is weaker than \ref{MC}, is not sufficient for the existence of a strictly increasing, mixture-preserving function.   
Propositions \ref{p:strict1} and \ref{p:strict2} together strengthen this result: the existence of a mixture-preserving multi-representation
(a condition stronger than \ref{MC}, and also stronger than \ref{Con} for the weak topology) is not sufficient either.

As we discussed in section~\ref{s:related lit}, \citet{DMO2004} consider mixture preorders on the set of probability measures on a compact metric space, and assume \ref{Cl} with respect to the narrow topology. Besides proving the existence of a mixture-preserving, and indeed expectational, multi-representation, they also prove in their Proposition 3 the existence of a strictly increasing expectational function. 
\citet{lG2017} uses this to prove the existence of a multi-representation by strictly increasing expectational functions. Our proof of Proposition \ref{p:strict1} is based on a similar technique.

\citet{oE2014} also considers probability measures on a compact metric space. He does not focus on multi-representations in our sense, but nonetheless gives conditions under which a preorder can be represented  by a set of strictly increasing functions 
in a different sense, which may have some advantages. We note that Evren's approach is essentially incompatible with ours (and with the one of \citeauthor{DMO2004}), insofar as his main continuity axiom, `open-continuity,' rarely holds when \ref{MC} does: a mixture preorder that satisfies both is either complete or symmetric.  

\subsection{Uniqueness}\label{s:unique}



Finally, we give a uniqueness result for mixture-preserving multi-representations. It is very similar to the uniqueness theorem of \citet{DMO2004}, but worked out in our setting of abstract mixture spaces.

Given a mixture space $M$, we let $M^*$ be the vector space of all 
real-valued mixture-preserving functions on $M$. Let $\CC \subset M^*$ be the subspace of constant functions. 
We give $M^*$ the topology of pointwise convergence: 
the coarsest topology such that for each $x \in M$, the function $M^* \to \RRR$ given by $f \mapsto f(x)$ is continuous. 
We write $\ov S$ for the closure of a subset $S$ of $M^*$. 

\begin{prop}\label{p:unique} Let $M$ be a mixture space. Two nonempty sets $\UU, \UU' \subset M^*$ represent the same preorder on $M$ 
if and only if $\ov{\cone{(\UU \cup \CC})} = \ov{\cone{(\UU' \cup \CC)}}$. 
\end{prop}

It is easy to check that if $\UU$ represents $\se$, then the subset of functions in $M^*$ that are increasing with respect to $\se$ is the unique {\em maximal} mixture-preserving multi-representation of $\se$. Proposition \ref{p:unique} is equivalent to the claim that the closure of the convex cone containing $\UU$ and the constant functions {\em is} this maximal multi-representation.

\appendix

\section{Independence and weak continuity}\label{s:WCon}


In this appendix, we clarify how our basic axioms, \ref{SI} and \ref{MC}, are related to others common in the literature on expected utility without completeness, 
as mentioned in section \ref{s:related lit}.

Let $M$ be a mixture space, and consider the following axioms for a preorder $\se$ on $M$.

\newlist{Ind}{enumerate}{1}
\setlist[Ind]{label=\textbf{Ind},ref={\text{Ind}}}

\begin{Ind}[align=parleft, leftmargin=!,itemsep=0pt,labelsep=14pt] 
        \item\label{Ind} For $x$, $y$, $z \in M$, and $\aa \in (0,1)$,
$x\se y \implies  x \aa z \se  y\aa z$.    
\end{Ind}

\newlist{WCon}{enumerate}{1}
\setlist[WCon]{label=\textbf{WCon},ref={\text{WCon}}}

\begin{WCon}[align=parleft, leftmargin=!,itemsep=0pt,labelsep=14pt] 
        \item\label{WCon}  For $x$, $y$, $z$, $w \in M$, $\Set{\aa \in [0,1]}{x\aa y \se  z\aa w}$ is closed. 
\end{WCon}
The first is the independence axiom of expected utility theory. The second is axiom P4 of \citet{SB1998}, 
and 
is called `weak continuity' by \citet{DMO2004}. 
Some relationships are clarified by 
the following lemma.
\begin{lemma}\label{l:WCon} Let $\se$ be a preorder on a mixture space $M$. 
\begin{enumerate}
[nosep, label = (\roman*)]
\item\label{l:WCon:WCon} \ref{WCon} $\implies$ \ref{MC}; 
\item\label{l:WCon:MC} \ref{MC} \& \ref{Ind} $\centernot\implies$ \ref{WCon}; 
\item\label{l:WCon:all} \ref{WCon} \& \ref{Ind} $\iff$ \ref{MC} \& \ref{SI}. 
\end{enumerate}
\end{lemma}
Thus \ref{MC} is weaker than \ref{WCon}, \ref{SI} is stronger than \ref{Ind}, and following \citet{SB1998}, we could have focused on 
the package of  \ref{Ind} and \ref{WCon} instead of  \ref{SI} and \ref{MC}. 
We have emphasized the latter combination partly because \ref{MC} seems simpler and more intuitive than \ref{WCon}, and partly because  
\ref{SI} is arguably the central idea of expected utility: if $M$ is a convex set of probability measures, 
\ref{SI} is necessary and sufficient for a preorder on $M$ to have an 
vector-valued 
expectational representation \citep[Lem. 4.3]{MMT2020}.

\newlist{HM}{enumerate}{1}
\setlist[HM]{label=\textbf{HM},ref={\text{HM}}}

\begin{rem}\label{r:MC} Intermediate between \ref{MC} and \ref{WCon} is the Herstein-Milnor axiom 
\begin{HM}[align=parleft, leftmargin=!,itemsep=0pt,labelsep=34pt] 
        \item\label{HM}  For $x$, $y$, $z\in M$, $\Set{\aa \in [0,1]}{x \aa y \se z}$  and $\Set{\aa \in [0,1]}{z \se x\aa y}$  are closed.  
\end{HM}
Since it is clear that $\ref{WCon}\implies\ref{HM}\implies\ref{MC}$, Lemma
\ref{l:WCon}\ref{l:WCon:all} 
shows that all three of these conditions are equivalent for mixture preorders (i.e.~assuming \ref{SI}). Such an equivalence between \ref{MC} and \ref{HM} was already noted by \citet{rA1962},  without proof. 
\end{rem}

\begin{proof}[{\bf Proof of Lemma~\ref{l:WCon}}]
\ref{l:WCon:WCon} Take $w = z$ in the statement of \ref{WCon}.

\ref{l:WCon:MC} 
Take $M = [0,1]$ and define $\se$ by $1 \s x \e y$ for all $x$,  $y \in [0,1)$.
This $\se$ is easily seen to satisfy \ref{MC} and \ref{Ind}, but not \ref{SI}.  
\citet[Lem. 1.2]{SB1998} is that \ref{WCon} \& \ref{Ind} $\implies$ \ref{SI}. It follows that $\se$ violates  \ref{WCon} (as one can check with $x=y=z=0$, $w=1$).

\ref{l:WCon:all} The 
left-to-right 
direction is immediate from \ref{l:WCon:WCon} and the result by Shapley and Baucells just mentioned. 
For the 
right-to-left direction, assume \ref{MC} and \ref{SI}. \ref{SI} obviously entails \ref{Ind}, so it remains to derive \ref{WCon}. 
It is possible to give a direct proof, using only the mixture space axioms. 
However, a shorter proof is available in terms of an efficient embedding. 
%
We emphasize that this involves no circularity, as \citet{SB1998} derived  the results concerning efficient embeddings that we presented in section~\ref{s:mspace} using only \ref{SI}, having first derived it from \ref{WCon} and \ref{Ind}; see note~\ref{n:SB1998}.

Assume, then, that $M \subset V$ is an efficient embedding, with $C$ the positive cone. 
By Proposition~\ref{p:cone}\ref{p:cone:MC}, 
whose proof does not depend on the present result, 
$C$ is algebraically closed. Consider the set 
$I=\Set{\aa\in[0,1]}{\aa x + (1-\aa)y \se \aa z + (1-\aa)w}$,  as in the statement of \ref{WCon}. 
Define $f(\aa)=\aa(x-z) + (1-\aa)(y-w)$. Thus $f$ maps $[0,1]$ onto the line segment $I'=[y-w,x-z]=\Set{\aa(x-z)+(1-\aa)(y-w)}{\aa\in[0,1]}$.  Since $C$ is convex, $I'\cap C$ is a (possibly empty) line segment; since $C$ is algebraically closed, this line segment, if not empty, contains its end points. But by \eqref{eq:mspace}, $\aa\in I\iff f(\aa) \in C$, so $I=f\inv(I'\cap C)$. It follows that $I$ is a closed interval, implying \ref{WCon}.
\end{proof}

\section{Weak dominance and Archimedean  structures}\label{s:weak dom}

Let $(M,\se)$ be a preordered mixture space. In this appendix we prove some general facts about weak dominance that we used in section \ref{s:CD}. Primarily, we show that weak dominance is a preorder on $\Graph$. 
This enables us to define the Archimedean structure $(\ArchCl,\geq)$ as in section \ref{s:weak Arch}: $\ArchCl$ consists of equivalence classes in $\Graph$ under the symmetric part of the weak dominance preorder.
While it is not difficult to check the preordering property directly, we proceed in a way that highlights a geometrical interpretation of the Archimedean structure: 
it  is closely related to the lattice of faces of the  positive  cone $C$ defined by an efficient embedding $M\subset V$ (cf.~section \ref{s:mspace}).  This was illustrated in Example~\ref{ex:arch}. 

Recall that a non-empty convex subcone $F\subset C$ is called a \emph{face} of $C$ if, for all $x,y\in C$, $x+y\in F\implies x,y\in F$. 
The set $\FF$ of faces is partially ordered by inclusion, and indeed it is a complete lattice.\footnote{See \citet{gB1973}, from which we take our simple definition of a face of a convex cone; it is compatible with the standard definition of the face of a convex set.}
This means in particular that, 
for any $v\in C$, there is a smallest face $\Phi(v)$  containing $v$.   
Let us say that $F\in\FF$ is \emph{regular} if $F$ is not the union of its proper subfaces: equivalently, $F=\Phi(v)$ for some $v\in C$. Let $\FF_r\subset \FF$ be the set of regular faces. 

\begin{prop}\label{p:faces}~
\begin{enumerate}[nosep, label =(\roman*)]
\item\label{p:faces:basic} For any $(x,y),(s,t)\in\Graph$, $(x,y)$ weakly dominates $(s,t)$ if and only if $\Phi(x-y)\supset \Phi(s-t)$.
\item\label{p:faces:preorder} Weak dominance is a preorder on $\Graph$.
\item\label{p:faces:iso} $(\ArchCl,\geq)$ is isomorphic to $(\FF_r, \subset)$ as a partially ordered set. 
\item\label{p:faces:min} {} Any $[(x,y)]\in\ArchCl$ is minimal if and only if $\Phi(x-y)=C$. In particular, $\ArchCl$ contains at most one minimal element.
\end{enumerate}
\end{prop}

\begin{proof}
For \ref{p:faces:basic},  
suppose that $(x,y)$ weakly dominates $(s,t)$. Then there exists $\aa\in(0,1)$ such that $\aa x+(1-\aa) t\se \aa y+(1-\aa)s$. Let $\ll=\frac{1-\aa}{\aa}$. It follows 
from \eqref{eq:CseVdef}
that 
$(x-y)-\ll(s-t)\in C$. 
At this point we appeal to 
\citet[Lemma 2.8]{gB1973}: $w\in \Phi(v)$ if and only if there exists $\ll>0$ such that $v-\ll w\in C$. (We note that Barker's lemma does not use his standing assumption of finite-dimensionality.) 
In our case, we find $s-t\in \Phi(x-y)$, and therefore 
$\Phi(s-t)\subset \Phi(x-y)$. 
The argument is reversible.

Part \ref{p:faces:preorder} now follows from the fact that `$\supset$' is a preorder on $\FF$. 


Now for part \ref{p:faces:iso}. 
It follows from part \ref{p:faces:basic} that $[(x,y)]\mapsto \Phi(x-y)$ is a well-defined, order-preserving, injective function $\ArchCl\to \FF_r$, and we just have to show it is surjective, i.e. that every regular face is of the form $\Phi(x-y)$ with $(x,y)\in\Graph$. Every regular face is of the form $\Phi(v)$, with $v\in C$, and, by \eqref{eq:CseVdef}, every such $v$ is of the form $\ll(x-y)$ with $\ll>0$ and $(x,y)\in\Graph$. Since every face containing $v$ contains $\tfrac1\ll v$, and vice versa, we find that $\Phi(v)=\Phi(x-y)$.    


For \ref{p:faces:min}, $C$ is the minimal  face of $C$
with respect to the preorder `$\subset$' (i.e.~it is set-theoretically the \emph{largest} face). So, if $\Phi(x-y)=C$, then certainly $C$ is a minimal regular face, and therefore $[(x,y)]$ is minimal. Conversely, if $[(x,y)]$ is minimal, then $\Phi(x-y)$ is a minimal regular face. It remains to show that, if there is a minimal regular face, then it is $C$. Suppose $\Phi(v)$ is a minimal regular face. 
Note that, for any $w\in C$, any face containing $\Phi(\tfrac12v+\tfrac12w)$ contains  $v+w$, and therefore contains  both $v$ and $w$. Therefore $\Phi(v)\subset \Phi(\tfrac12v+\tfrac12w)$.  Since $\Phi(v)$ is minimal regular,  $\Phi(v)=\Phi(\tfrac12v+\tfrac12w)\ni w$.  
That is, $\Phi(v)$ contains every $w\in C$; so $\Phi(v)=C$. 
\end{proof}

\section{Proofs of auxiliary results}\label{s:aux proofs}

\begin{proof}[\bf Proof of Proposition~\ref{p:dim}]
For the first claim, suppose that $(M,\se)$ is a preordered mixture space of countable dimension. We appeal to some results from section \ref{s:discuss}, the proofs of which do not depend on this one.  In the terminology of section \ref{s:mspace}, let $M\subset V$ be an efficient embedding, 
with $C$ the positive cone. 
Proposition \ref{p:cone}\ref{p:cone:dim} shows that $V$ has countable dimension. Therefore its subspace $\Span(C)$ has countable dimension.  Corollary \ref{c:CD dim} then tells us that $\se$ satisfies \ref{CDom} (note that $\Span(C)$ is a cofinal subspace of itself). 

The second claim, that the converse does not hold, 
even for mixture preorders that satisfy \ref{MC}, 
is illustrated by Examples \ref{ex:combo} and \ref{ex:SDnotAr}.  
\end{proof}

\begin{proof}[\bf Proof of Proposition~\ref{p:cone}.] 

For \ref{p:cone:dim},  
let $A\subset M$  be nonempty.  Fix any $a_0\in A$ and let  $A' = \Set{a-a_0}{a\in A\setminus\{a_0\}}$. Since 
$M \subset V$ is an efficient embedding, $V = \Span(M-M) = \Span(M - \{a_0\})$. 
Thus, $A'$ is a basis for $V$ if and only if it is linearly independent and maximal among linearly independent subsets of $M-\{a_0\}$.  We claim that $A'$ is linearly independent if and only if $A$ is mixture independent. It follows that $A'$ is a basis for $V$ if and only if $A$ is a maximal mixture-independent subset of $M$. Since $|A'|=|A|-1$, it follows that the vector-space dimension of $V$ equals the mixture-space dimension of $M$.

To prove the claim, first suppose that $A$ is \emph{not} mixture independent.  There must be nonempty 
$A_1, A_2\subset A$ such that $A_1\cap A_2=\emptyset$ but $M(A_1)\cap M(A_2)\neq \emptyset$. Given the embedding of $M$ into $V$, $M(A_1)$ equals the convex hull of $A_1$; it consists of all convex combinations of elements of $A_1$. Since $M(A_1)\cap M(A_2)\neq\emptyset$, 
there is an equality between two convex combinations of the form
\[
\textstyle{\sum_i \aa_i x_i = \sum_i \bb_i y_i}
\]
with 
$\aa_i,\bb_i\in[0,1]$, $x_i\in A_1$, $y_i\in A_2$, and $\sum \aa_i= \sum\bb_i=1$. But then we also have
\[
\textstyle{\sum_i \aa_i (x_i-a_0) = \sum_i \bb_i (y_i-a_0)}
\]
showing that $A'$ is linearly dependent.  
Conversely, suppose that $A'$ is linearly dependent. Then there are disjoint, finite $A_1, A_2\subset A\setminus\{a_0\}$ and an equation of the form 
\[
	\textstyle{\sum_i \ll_i(a_i - a_0) = \sum_i \mu_i(b_i - a_0)}
\]
where at most one of the sums is empty (in which case it is zero), 
with all $\ll_i, \mm_i >0$, 
$a_i \in A_1$, $b_i \in A_2$.  
Without loss of generality, we can assume that $\ll\coloneqq \sum_i \ll_i\geq \sum_i \mm_i\eqqcolon \mu$, 
so that $A_1$ is nonempty. 
Moving all terms involving $a_0$ to the right-hand side, and dividing by $\ll$, we have 
\[
\textstyle{\sum_i \frac{\ll_i}{\ll} a_i  = \sum_i \frac{\mu_i}{\ll} b_i+\frac{\ll-\mu}{\ll} a_0.}
 \]
This shows that $M(A_1)\cap M(A_2\cup\{a_0\})\neq\emptyset$. Therefore $A$ is not mixture 
independent.

Now for part \ref{p:cone:MC}. Suppose first that $\se$ satisfies \ref{MC}.  
Let $(v,w]\subset C$, so that, by \eqref{eq:CseVdef}, $z\seV 0$ for every $z\in(v,w]$. To show that $C$ is algebraically closed, we have to show $v \in C$. Suppose first that $z\eV0$ for some $z\in (v,w]$. Then $-z \eV 0$, so, by \eqref{eq:CseVdef} again, $-z\in C$. Let $z'=\tfrac12 v+\tfrac 12z\in(v,w]$. We have $v=2z'-z$. Since both $z'$ and $-z$ are in $C$, and $C$ is a convex cone,  it follows that $v\in C$, as desired.  
We are thus reduced to the case where $z\sV 0$ for every $z\in (v,w]$.  

Now we claim that there exists $\ll_0>0$ and $x_0,x_1,x_2\in M$ such that $v=\ll_0(x_1-x_0)$ and $w=\ll_0(x_2-x_0)$. 
Since $M\subset V$ is an efficient embedding, 
using \eqref{eq:Vspan} we can 
write $v = \ll (x-y)$ and $w = \mm(s-t)$ for some $\ll, \mm >0$ and $x,y,s,t \in M$.  
Set $\bb=\ll/(\ll+\mm)$, so $1-\bb=\mm/(\ll+\mu)$. 
The claim is easily verified with 
\[\ll_0=\ll+\mm,\quad 
x_0=\bb y+(1-\bb)t,\quad x_1=\bb x+(1-\bb)t,\quad x_2=\bb y+(1-\bb)s.\] 

Any $z\in(v,w]$ can be written as $z=(1-\aa)v+\aa w$, with $\aa\in (0,1]$. It follows that 
\[z=(1-\aa)\ll_0(x_1-x_0)+\aa\ll_0(x_2-x_0)=\ll_0((1-\aa)x_1 +\aa x_2 -x_0).
\] Since, as in the first step, $z\sV 0$, it follows that $(1-\aa)x_1+\aa x_2\sV x_0$. Then, by \eqref{eq:mspace}, 
$(1-\aa)x_1+\aa x_2\s x_0$. This holds for all $\aa\in(0,1]$, so 
 \ref{MC} gives us $x_1\se x_0$. Therefore, by \eqref{eq:CseVdef}, $v=\ll_0(x_1-x_0)\in C$.

Conversely, suppose that $C$ is algebraically closed. To show that $\se$ satisfies \ref{MC}, 
suppose that $\aa x+(1-\aa) y \s z$ for all $\aa\in(0,1]$. Then by \eqref{eq:mspace}, $\aa(x-z)+(1-\aa)(y-z)\in C$ for all such $\aa$. Since $C$ is algebraically closed, it follows that $y-z\in C$. 
By \eqref{eq:mspace}, $y\se z$, validating \ref{MC}.

For \ref{p:cone:MR}, let $V$ have the weak topology. Suppose first that $\se$ has a 
mixture-preserving multi-representation 
$\UU$. Then \eqref{eq:intersection} presents $C$ as the intersection of 
closed sets, 
so it is closed. 

Conversely, 
suppose that $C$ is closed. 
If $C = V$, then by \eqref{eq:CseVdef} $\se$ is the indifference relation, which has a mixture-preserving multi-representation consisting of 
a single constant function. 
Assume then $C \neq V$. 
The weak topology on $V$ is locally convex, 
so by the strong separating hyperplane theorem 
\citep[Cor. 5.84]{AB2006}, 
for any $v \notin C$, there exists a linear functional $L_v \colon V \to \RRR$ such that $L_v(C) \subset [0, \infty)$ and $L_v(v) <0$. Let $\LL = \Set{L_v}{v \notin C}$. Then by \eqref{eq:mspace}, 
\[
	x \Geq y \iff  x-y \in C \iff L(x) \geq L(y) \text{ for all } L \in \LL. 
\]
It follows that the restriction of $\LL$ to $M$ is a  mixture-preserving multi-representation of $\se$. 
\end{proof}

\begin{proof}[{\bf Proof of Proposition \ref{p:klee}}]
We show that the cone $K$ defined in Example~\ref{ex:klee} is algebraically closed but not 
closed (recall that $V$ has the weak topology).

As a first step, we show that, for any finite, non-empty $A\subset B_1$, the subcone $K\cap\Span(A\cup\{b_0\})$ of $K$ is algebraically closed.  
Any convex cone generated by finitely many elements is algebraically closed \citep[see e.g.][G.1.6, Thm.~1]{eO2007}, so it suffices to prove 
\begin{equation}\label{eq:klee-a-pf}
K\cap\Span(A\cup\{b_0\}) = \cone\Set{y_{A'}+b_0}{A' \neq \varnothing, A'\subset A}.
\end{equation}
The inclusion of the right-hand side in the left is obvious. Conversely, suppose $v$ is a member of the left-hand side. We may assume $v \neq 0$. 
Since $v \in K$, it may be written 
\begin{equation}\label{eq:ee}
v = \sum_{k=1}^n \ll_k (y_{A_k}+b_0)
\end{equation}
where $n$ is a positive integer, each coefficient $\ll_k$ is strictly positive, and each $A_k$ is a finite, nonempty subset of $B_1$. It follows that $v$ is a linear combination, with all coefficients strictly positive, of every member of $\bigcup_{k=1}^nA_k \cup \{b_0\}$. Since $v \in \Span(A\cup\{b_0\})$, this is only possible if $A_k \subset A$ for each $k$. Therefore \eqref{eq:ee} presents $v$ as a  member of the right-hand side of \eqref{eq:klee-a-pf}.

We can now show that $K$ itself is algebraically closed. 
Suppose given a half-open line segment $(v_0,v_1]\subset K$; we have to show $v_0\in K$. We can find a finite set of basis elements $A\subset B_1$ such that $v_0,v_1\in \Span(A\cup\{b_0\})$, and therefore such that $(v_0,v_1]\subset \Span(A\cup\{b_0\})$. Since $K\cap \Span(A\cup\{b_0\})$ is algebraically closed, it contains $v_0$; therefore $v_0\in K$, as desired.

Finally, we show that $K$ is not 
 closed. 
In this proof, let 
$\ov K$ 
denote the 
closure of $K$.
Note that $b_0\notin K$; we show that $b_0$ is nonetheless in $\ov K$. 
Suppose for a contradiction $b_0 \notin  \ov K$. 
By the strong separating hyperplane theorem 
there exists a linear functional 
$f\colon V\to \RRR$ such that $f(b_0)<0$ but 
$f(K) \subset [0,\infty)$. 
Now, since $B_1$ is uncountable, 
there exists some $n\in\NNN$ for which there are infinitely many $b\in B_1$ with $f(b)<n$. Let $A$ be a nonempty, finite set of such $b$. 
Then $f(y_A)<|A|^{-2}\sum_{b\in A} n= n/|A|$. 
Therefore $f(y_A+b_0)<f(b_0)+n/|A|$. Since $|A|$ may be chosen to be arbitrarily large, and $f(b_0)<0$,  
we can find some $y_A$ such that $f(y_A + b_0) <0$, contrary to $f(K) \subset [0, \infty)$. We conclude that $b_0\in\ov K$.
\end{proof}

\begin{proof}[{\bf Proof of Proposition \ref{p:count}}]
Let $C$ be an algebraically closed convex subset of a vector space $V$. We may assume $C$ is nonempty; we want to show it is 
closed when $V$ is endowed with the weak topology.

First consider the case when $\dim V$ is finite. 
The weak topology on $V$ is 
then 
the same as the Euclidean topology. 
The following argument is based on \citet[\S11A(c)]{rH1975}. 
We use the fact that $C$, like any convex subset in a finite-dimensional vector space,  has a non-empty relative interior $\rInt C$ \citep[Lemma 7.33]{AB2006}. This is an open subset of $\aff C$. 
Translating $C$, we can assume that $0\in\rInt C$, in which case $\aff C=\Span C$. 
Let $x$ be in the closure of $C$, which is contained in $\aff C$. For any $\aa\in(0,1)$, $X=-\frac{1-\aa}{\aa}\rInt C$ is open in $\aff C$, so $x+X$ contains a point $x'\in C$. Then
 \[
\aa x\in \aa (x'-X)=\aa x'+ (1-\aa)\rInt C\subset C. 
\]
Thus $(x,0]\subset C$. Since $C$ is algebraically closed, $x\in C$; thus $C$ is 
closed.

Now suppose $V$ has countable dimension. By definition, a subset $X$ of $V$ is closed in the {\em finite} topology on $V$ if and only if $X\cap W$ is closed in the Euclidean topology in every finite-dimensional subspace $W$ of $V$. 
Since, for each finite-dimensional $W\subset V$, $C\cap W$ is algebraically closed, the preceding argument shows that $C$ is closed in the finite topology. 
By a result due to \citet{vK1953}, but stated more fully in \citet{KK1963}, the finite topology on a countable dimensional vector space makes it a locally convex topological vector space.  
By another version of the strong separating hyperplane theorem \citep[Cor. 5.80]{AB2006}, 
$C$ is the intersection of 
half-spaces that are closed in the weak topology. 
$C$ itself is therefore closed in the weak topology. 
\end{proof}

The proof of Proposition~\ref{p:misc} will use the following observation, given an efficient embedding $M\subset V$ of a preordered mixture space.

\begin{lemma}\label{l:wdom} Suppose given $(s,t)\in \Graph$, 
$x,y\in M$, 
and $\mu>0$.
The following are equivalent:
\begin{enumerate}
[nosep, label=(\roman*)]
\item\label{l:wdom:ll} There exists $\ll>0$ such that $\ll(x-y)\seV \mu(s-t)$.
\item\label{l:wdom:wdom} We have $(x,y)\in \Graph$, and $(x,y)$ weakly dominates $(s,t)$.
\end{enumerate}
\end{lemma}

\begin{proof} 
We repeatedly use facts \eqref{eq:mspace} and \eqref{eq:CseVdef} about efficient embeddings. 
Suppose \ref{l:wdom:ll} holds. We have $(s,t)\in \Graph\implies s\se t\implies \mu(s-t)\seV 0\implies \ll(x-y)\seV0\implies x\se y \implies (x,y)\in \Graph$. 
Rearranging the inequality in \ref{l:wdom:ll}, and setting $\aa=\ll/(\ll+\mu)$, 
we find $\aa x+(1-\aa)t\se \aa y+(1-\aa)s$. 
Therefore $(x,y)$ weakly dominates $(s,t)$. Thus \ref{l:wdom:wdom} holds. Conversely, given \ref{l:wdom:wdom}, 
we have 
$\aa x+(1-\aa)t\se \aa y+(1-\aa)s$ for some $\aa\in(0,1)$. Rearranging, we obtain $\ll(x-y)\seV \mu(s-t)$ with $\ll=\aa\mu/(1-\aa)$. Thus \ref{l:wdom:ll} holds.
\end{proof}

\begin{proof}[{\bf Proof of Proposition~\ref{p:misc}}]
For~\ref{p:misc:SD}, it is a standard result that the algebraic interior of a convex cone consists of its order units; see e.g. \citet[Lemma 1.7]{AT2007}. The proof of~\ref{p:misc:SD} essentially translates this fact into a result about $M$ itself. We will rely on the basic facts \eqref{eq:mspace} and \eqref{eq:CseVdef} about efficient embeddings without further comment.

Suppose \ref{SDom} holds with respect to some $(x,y)\in \Graph$. Let $v=x-y\in C$.  We note that, since $C$ is a convex cone, 
$\aff(C)=\Span(C)=C-C$.  
Thus, given $w\in\aff(C)$, we can write $w=w_1-w_2$ with $w_1,w_2\in C$. 
Since $w_2\in C$,  we also have $w_2=\mu(s-t)$ for some $\mu>0$ and $(s,t)\in \Graph$.  By \ref{SDom}, $(x,y)$ weakly dominates $(s,t)$. So there exists, by Lemma \ref{l:wdom}, some $\ll>0$ such that $\ll v=\ll(x-y)\seV \mu(s-t)=w_2$. Therefore $v-\tfrac1\ll w_2\in C$. Since also $\tfrac1\ll w_1\in C$, we find that $v+\tfrac1\ll w_1-\tfrac1\ll w_2=v+\tfrac1\ll w\in C$. Since $C$ is convex, we deduce $[v,v+\tfrac1\ll w]\subset C$. Since $w\in\aff(C)$ was arbitary, this shows $v$ is in the relative algebraic interior $\rai(C)$.

Conversely, suppose that $\rai(C)$ is nonempty. 
Fix $v\in\rai(C)$; then $v=\ll(x-y)$ for some $\ll>0$ and $x\se y$. 
Given any $(s,t)\in \Graph$, we have $t-s\in -C\subset\aff(C)$. For some $\epsilon>0$, we must have $v+\epsilon(t-s)\in C$, so 
$\ll(x-y)\seV \epsilon(s-t)$.
By Lemma \ref{l:wdom}, we have $(x,y)\in \Graph$ and $(x,y)$ weakly dominates $(s,t)$. Therefore this $(x,y)$ weakly dominates every $(s,t) \in \Graph$, so \ref{SDom} holds.

For~\ref{p:misc:CD}, suppose \ref{CDom} holds, so that every $(s,t)\in \Graph$ is weakly dominated by an element of some countable set $D \subset \Graph$. Let
$S = \Set{n(x-y)}{n \in \NNN, (x,y) \in D}\subset C$. 
Since $D$ is countable, so is $S$.
We claim $S$ is cofinal in $C$.
Let  $w \in C$. 
We can write $w=\mu(s-t)$ with $\mu>0$, 
$s\se t$. 
Some $(x,y)\in D$ weakly dominates $(s,t)$. Therefore, by Lemma \ref{l:wdom},  there exists $\ll>0$ with $\ll(x-y)\seV \mu(s-t)=w$.  Choose an integer $n>\ll$. Then $n(x-y)\seV\ll(x-y)\seV w$. Since $n(x-y)\in S$,  $S$ is cofinal in $C$.  

Conversely, suppose that $S$ is a countable set, cofinal in $C$.  
For each $v \in S$,  
we can choose $\ll_v>0$ and  
$x_v, y_v \in M$ with $x_v \se y_v$ 
such that  $v=\ll_v(x_v-y_v)$. 
Let $D=\Set{(x_v,y_v)}{v\in S}$. 
Since $S$ is countable, so is $D$. To prove \ref{CDom}, we show that every $(s,t)\in \Graph$ is weakly dominated by an element of $D$.
Since $s \se t$, we have $s-t\in C$. 
Since $S$ is cofinal, there exists $v \in S$ such that  $v\seV s-t$. It follows from Lemma \ref{l:wdom} that $(x_v,y_v)$ weakly dominates $(s,t)$. 
\end{proof}

\begin{proof}[{\bf Proof of Corollary \ref{c:CD dim}}]
By Proposition \ref{p:misc}\ref{p:misc:CD}, it suffices to show that there is a countable 
set  cofinal in $C$ if and only if there is a countable-dimensional subspace cofinal in $\Span(C)$.

Suppose $S\subset C$ is countable and cofinal.  Let $Z=\Span(S)$. Because $C$ is a convex cone, any $v\in\Span(C)$ can be written in the form $v=x-y$ with $x,y\in C$. There is some $s\in S$ such that $s\seV x$; but then $s\seV v$. 
Since $s\in Z$, $Z$ is cofinal in $\Span(C)$. 
It has countable dimension since $S$ is countable.

Conversely, suppose a countable-dimensional subspace $Z$ is cofinal in $\Span(C)$. 
Let $b_1,b_2,\ldots$ be a countable (finite or infinite) basis for $Z$. Since $b_i\in \Span(C)$,  it can be written as $x_i-y_i$ with $x_i,y_i\in C$. Note that $x_i\seV b_i$. Let $S$ consist of all linear combinations of the $x_i$ with non-negative integer coefficients; it is a countable subset of $C$. Let $v\in C$. There exists $z\in Z$ such that $z\seV v$. We can write $z$ as a finite sum $z=\sum_i \ll_ib_i$, for some $\ll_i\in\RRR$. If $\ll$ is a positive integer greater than all the $\ll_i$, then $S\ni \sum_i \ll x_i\seV z\seV v$. Therefore $S$ is cofinal in $C$.
\end{proof}


\begin{proof}[{\bf Proof of Lemma~\ref{l:top}}]
The first claim, at least, is well-known; \citet{BH2016}, for example, provide two proofs of the first implication. But we give the short proofs for convenience.

To show $\ref{CMR}\implies\ref{Cl}$,  suppose $\UU$ is a continuous mixture-preserving multi-representation of $\se$.  For each $u\in \UU$, define $\tilde u\colon M^2\to \RRR$ by $\tilde u(x,y)=u(x)-u(y)$. This $\tilde u$ is continuous,
and $\Graph=\bigcap_{u\in \UU} \tilde u\inv([0,\infty))$. Thus $\Gamma_\se$ is the intersection of closed sets, so \ref{Cl} holds.

To show $\ref{Cl}\implies\ref{Con}$,  assume that $\Graph$ is closed.  Let $x \in M$. The map $f_x \colon M \to M^2$ given by $f_x(y) = (y,x)$ is continuous. Therefore, $\Set{y}{y \se x} = f_x\inv(\Graph)$ is closed. 
A similar argument shows that $\Set{y}{x \se y}$ is closed. Hence \ref{Con} holds.

For the second claim of the lemma, suppose $M$ is a mixture space and the maps $f_{x,y}$ are continuous. To show $\ref{Con}\implies\ref{MC}$, 
suppose that $\se$ is continuous.  
Suppose that $x \aa y \s z$ for all $\aa \in (0,1]$. Since $\Set{w}{w \se z}$ is closed, so is $f_{x,y}\inv(\Set{w}{w \se z})$. 
The latter contains $(0,1]$, so it also contains $0$. Thus 
$y \se z$, establishing \ref{MC}. 
\end{proof}

The next lemma records some basic facts about the weak topology that will be used in the proof of Proposition~\ref{p:klee2}.

\begin{lemma}\label{l:MS}
Let $M_1$ and $M_2$ be mixture spaces, each with the weak topology.
\begin{enumerate}
[nosep, label=(\roman*)]
\item\label{MS:con} Suppose $f \colon M_1 \to M_2$ is mixture-preserving. Then $f$ is continuous.
\item\label{MS:prod} The weak topology on $M_1\times M_2$ equals the product topology.%
\footnote{Here $M_1\times M_2$ is a mixture space with respect to the component-wise mixing operation: $(x_1,x_2)\aa(y_1,y_2)=(x_1\aa y_1,x_2\aa y_2)$.} 
\item\label{MS:sub} If $M_1$ is a mixture subspace of $M_2$,  then it is a topological subspace. 
\item\label{MS:vsub} If $M_2$ is a vector space and $M_1 \subset M_2$ is a linear subspace, 
then $M_1$ is closed in $M_2$.
\end{enumerate}
\end{lemma}

\begin{proof} 
 \ref{MS:con} By definition of the weak topology on $M_2$, a function $f\colon X\to M_2$ from an arbitrary topological space $X$ is continuous if and only if $g\circ f$ is continuous for every mixture-preserving $g\colon M_2\to\RRR$. Our $f\colon M_1\to M_2$ is mixture preserving, so $g\circ f$ is mixture-preserving, and therefore continuous on $M_1$.

\ref{MS:prod} 
 The weak topology on $M_1\times M_2$ is the coarsest one such that every mixture-preserving $f\colon M_1\times M_2\to \RRR$ is continuous. 
The product topology is the coarsest one such that 
the projections $\pi_i$ of $M_1\times M_2$ onto $M_i$ are continuous. 
Equivalently, it is the coarsest one such that 
for all mixture-preserving $f_1\colon M_1\to\RRR$ and $f_2\colon M_2\to\RRR$, the function
$f_1\circ\pi_1+f_2\circ\pi_2\colon M_1\times M_2\to\RRR$ is continuous. Since the latter function is clearly mixture-preserving, it suffices to show that (conversely) every mixture-preserving $f$ is of this form.  

 Fix $z_1\in M_1$ and $z_2\in M_2$.
 For $x_i\in M_i$ define $f_1(x_1)= f(x_1,z_2)$ and $f_2(x_2)= f(z_1,x_2)-f(z_1,z_2)$. It is easy to check that $f_1,f_2$ so defined are mixture preserving. Moreover, using the mixture-preservation property of $f$,  
\begin{align*} 
f_1(x_1)+f_2(x_2)-f(x_1,x_2) &=
f(x_1,z_2)+f(z_1,x_2)-(f(z_1,z_2)+f(x_1,x_2))\\
&=
2f(x_1\tfrac12 z_1,z_2\tfrac12 x_2)-
2f(z_1\tfrac12 x_1,z_2\tfrac12 x_2))
\\
&=0.
\end{align*}
Therefore $f_1\circ \pi_1+f_2\circ\pi_2=f$, as desired.

\ref{MS:sub}
The claim is that the weak topology on $M_1$ coincides with the subspace topology 
inherited from $M_2$. 
The restriction to $M_1$ of a mixture-preserving function on $M_2$ is mixture preserving;  
it follows that the subspace topology on $M_1$ is contained in its weak topology. To show the converse, it suffices to show that any 
mixture-preserving $M_1\to \RRR$ extends to a mixture-preserving function $M_2\to\RRR$. To prove this using standard facts from linear algebra, we can first embed $M_2$ as a convex set in a vector space $V$ (see section \ref{s:mspace}); thus $M_1$ is also a convex subset of $V$. Any mixture-preserving function $f\colon M_1\to\RRR$ extends to an affine (i.e. linear plus constant) function on $V$; 
the restriction of this affine function to $M_2$ is a mixture-preserving extension of $f$.

\ref{MS:vsub}
For any $x\in M_2\setminus M_1$, there is a linear (hence mixture-preserving) function $g\colon M_2\to\RRR$ such that $g(M_1)=\{0\}$ and $g(x)=1$. Then $g\inv((0,\infty))$ is an open neighbourhood of $x$ disjoint from $M_1$. Thus $M_2\setminus M_1$ is open and $M_1$ is closed in $M_2$.
\end{proof}

\begin{proof}[{\bf Proof of Proposition~\ref{p:klee2}}]
We first show that the mixture preorder $\se$ defined in Example~\ref{ex:klee2} is continuous.  Fix $(v,w)\in M$.  
Let $U = \Set{(x,y)}{(x,y)\se (v,w)}$ and $L = \Set{(x,y)}{(v,w)\se (x,y)}$. We need to show that $U$ and $L$ are closed in $M$, 
which has the weak topology. The two cases are similar, so we consider the former. 

Let $K_v = K\cap V_v$. 
Define a 
function $f \colon M \to V \times V$ by $(x,y) \mapsto (x,y) - (v,w)$. 
It follows from \eqref{eq:klee2} that
$U = f\inv(\{0\} \times K_v)$. 
Give $V\times V$ the weak topology.  Since $f$ is mixture-preserving,  Lemma~\ref{l:MS}\ref{MS:con} tells us that $f$ is continuous. So, to show that $U$ is closed, it suffices to show that $\{0\}\times K_v$ is closed in $V\times V$. 

In the first step of proving Proposition~\ref{p:klee} we showed that $K_v$, that is, $K\cap \Span( A_v\cup \{b_0\})$,   is an algebraically closed convex cone. Thus $\{0\} \times K_v$ is an algebraically closed convex subset of $\{0\} \times V_v$. Since $V_v$, and hence $\{0\} \times V_v$, is a finite-dimensional vector space, 
Proposition~\ref{p:count} implies that 
$\{0\} \times K_v$ is 
closed in the weak topology on $\{0\} \times V_v$.
 
By Lemma \ref{l:MS}\ref{MS:sub}, $\{0\}\times V_v$, with the weak topology, is a topological subspace of $V\times V$. Moreover, it is a closed subspace,  by Lemma \ref{l:MS}\ref{MS:vsub}. In summary, $\{0\}\times K_v$ is closed in a closed subspace of $V\times V$; therefore it is closed in $V\times V$.

We now show that $\Graph$ is not closed in $M\times M$.  Note that $z=(0,b_0;0,0)$ is an element of $M\times M$, but not of $\Graph$. It suffices to show that $z$ is in the closure of $\Graph$ in $M\times M$. 
Therefore, it suffices to find a net $(z_\aa)$ in $\Graph$ converging to $z$ in $M\times M$.
Here $M$ has the weak topology and $M\times M$ has the resulting product topology.
Similarly, 
give $V$ the weak topology, 
and $V^2\times V^2$ the product topology. By Lemma \ref{l:MS}\ref{MS:prod}, both these product topologies are again the weak topologies; Lemma \ref{l:MS}\ref{MS:sub} then implies that $M\times M$ is a topological subspace of $V^2\times V^2$.
So it will suffice that $(z_\aa)$ converges to $z$ in $V^2\times V^2$.

Recall that $b_0$ is in the 
closure $\ov K$ of $K$ in $V$,
as proved  as the last step in the proof of Proposition~\ref{p:klee}.
Let $(y_\aa)$ be a net in $K$ converging to $b_0$. Note that, by definition, $K\subset\cone(B)=V^+$. Therefore  each $y_\aa$ can be written as $y_\aa=x_\aa+\ll_\aa b_0$, with $x_\aa\in\cone(B_1)$ and $\ll_\aa\geq 0$. Note $x_\aa\in V^+$ and $y_\aa\in V_{x_\aa}$, so $(x_\aa,y_\aa)$ is in $M$. Moreover, by \eqref{eq:klee2}, $(x_\aa,y_\aa)\se (x_\aa,0)$. Therefore 
$z_\aa\coloneqq(x_\aa, y_\aa;  x_\aa,0)$ is in $\Graph$.  

Now, \emph{any} element of $V$ can be written uniquely in the form $y=x+\ll b_0$ with $x\in\Span(B_1)$ and 
$\ll\in\RRR$.
Define a linear map $f\colon V\to V^2\times V^2$ by $f(y)=(x,y;x,0)$. 
Note $z_\aa=f(y_\aa)$. 
Since,  by Lemma \ref{l:MS}\ref{MS:con}, $f$ is continuous, we have $\lim_\aa z_\aa =f(b_0)=z$. 
\end{proof}


\begin{proof}[{\bf Proof of Proposition~\ref{p:strict1}}]
It is obvious that a preorder satisfying \ref{SMR} satisfies \ref{MR} and admits a strictly increasing mixture-preserving function. (Note that we require multi-representations to be nonempty.) 
Conversely,
let $u' \colon M \to \RRR$ be 
mixture-preserving and 
strictly increasing, and $\UU$ be a mixture-preserving multi-representation. 
Let $\UU' = \Set{u' + nu}{n \in \NNN, u \in \UU}$. First, note that for any $n \in \NNN$ and $u \in \UU$, $u' + nu$ is strictly increasing. Now suppose that $u'(x) + nu(x) \geq u'(y) + nu(y)$ for all $n \in \NNN$, $u \in \UU$. Since, for each $u$, $n$ can be arbitrarily large, we must have $u(x) \geq u(y)$. Since $\UU$ is a multi-representation, we find $x \se y$, so $\UU'$ is a mixture-preserving multi-representation containing only strictly increasing functions.
\end{proof}

\begin{proof}[{\bf Proof of Lemma~\ref{l:card}}]
Let $M\subset V$ be an efficient embedding, with positive cone $C \subset V$. 
Suppose given a mixture-preserving multi-representation $\UU$. For each $u\in \UU$, let $\tilde u$ be its extension to an affine function $V\to\RRR$, and let $A_u$ be the open half-space $A_u= \Set{v\in V}{\tilde u(v)<\tilde u(0)}$.  It follows from \eqref{eq:intersection} that $\AA=\{A_u : u\in\UU\}$ is an open cover of $V\setminus C$, in the weak topology  
on $V$. 

Consider first the case where $\dim M$ is finite,   and hence, by Proposition \ref{p:cone}\ref{p:cone:dim}, $\dim  V$ is finite.  Then the weak topology on $V$ coincides with the Euclidean topology, and 
$V$ is a second-countable topological space, as is its topological subspace 
$V \setminus C$.
By Lindel\"of's lemma, 
$\AA$ contains a countable subcover $\AA'$. We can write 
$\AA'=\Set{A_u}{u\in\UU'}$ 
for some countable subset $\UU'\subset\UU$.  Then
\begin{equation}\label{eq:CVintersection'}
	C = \bigcap_{u \in \UU'}\Set{v\in V}{\tilde u(v)\geq \tilde u(0)}.
\end{equation}
It follows from \eqref{eq:intersection} that  $\UU'$ is a mixture-preserving multi-representation of $\se$. Finally we note that $|\UU'| =\aleph_0\leq \max(\aleph_0, \dim M)$.  

Now suppose $\dim M = \dim V = \kappa$ for some infinite cardinal $\kappa$. Let $B$ be a basis of $V$, and let $\PP$ be the set of finite subsets of $B$; note that $|\PP|= \kappa$. For each $P\in\PP$,  $\AA_P\coloneqq \{A_u\cap\Span P:u\in\UU\}$ is an open cover of $\Span P \setminus C$ 
in the weak topology on $\Span P$. 
As in the previous paragraph, it contains a countable subcover $\AA'_P$, which we can write in the form $\AA'_P=\{A_u\cap \Span P: u\in \UU'_P\}$, with $\UU'_P\subset\UU$ countable. 
Let $\UU'=\bigcup_{P\in\PP} \UU'_P$.
Choose any $v\in V\setminus C$. It is in $\Span P$ for some $P$, and therefore it is in $A_u$ for some $u\in\UU'$. So $\UU'=\{A_u:u\in\UU'\}$ is an open cover of $V\setminus C$. For the same reason as before, $\UU'$ is  a mixture-preserving multi-representation of $\se$. 
Finally, since $|\PP|=\kappa$ and each $\UU'_P$ is countable, $|\UU'|=\kappa\leq \max(\aleph_0, \dim M)$. 
\end{proof}

\begin{proof}[{\bf Proof of Proposition~\ref{p:strict2}}]
For \ref{p:strict2:count-a}, 
assume that $\dim M$ is countable and let $\se$ be a mixture preorder on $M$ that has a mixture-preserving multi-representation; we have to show that it has one using only strictly increasing functions. Let $M\subset V$ be an efficient embedding, so, by Proposition \ref{p:cone}\ref{p:cone:dim},  $\dim V$ is countable. Since $V = \Span M$, we can pick a (finite or countably infinite) basis $B=\{v_1,v_2,\ldots\}\subset M$ of $V$. 
By Lemma~\ref{l:card},  $\se$ has a finite or countably infinite 
mixture-preserving multi-representation $\UU= \{u_1,u_2,\ldots\}$. 
Let $\tilde u_i\colon V\to \RRR$ be the unique extension of $u_i$ to an affine function; thus 
$L_i \coloneqq  \tilde u_i - \tilde u_i(0)$ is a linear functional on $V$. 
Rescaling the $u_i$ as necessary,  we can assume $|L_i(v_j)| \leq 1$ whenever $j \leq i$. We define a mixture-preserving function $u$ on $M$ by 
\[u(x) = \sum_{i=1}^{|\UU|} 2^{-i}L_i(x). 
\]
This is clearly well-defined when $|\UU|$ is finite.   
If $|\UU|$ is infinite, note that every $x\in M$ can be written in the form 
$x = \sum_{j=1}^{|B|} c_j v_j$, with finitely many nonzero $c_j\in\RRR$.    
It follows that $|L_i(x)| \leq 
\sum_{j=1}^{|B|} |c_j||L_i(v_j)|\leq
 \sum_{j=1}^{|B|} |c_j|$, for all sufficiently large $i$. 
Therefore the sum defining 
$u(x)$ 
is absolutely convergent, making 
$u$ a well-defined mixture-preserving function. 
It is also strictly increasing.    By Proposition \ref{p:strict1}, $\se$ has a mixture-preserving multi-representation using only strictly increasing functions.

For part \ref{p:strict2:uncount-a}, we show that the mixture preorder defined in Example~\ref{ex:no strict} satisfies \ref{MR} but not \ref{SMR}.


That preorder was defined by a mixture-preserving multi-representation, 
so it satisfies \ref{MR}.  We show that it does not admit any strictly-increasing function $M\to \RRR$. Suppose for contradiction that $u$ is such a function. In the notation of the example, for each $\aa<\kk$, define $f(\aa) = -u(v_\aa)$. 
Given $\aa < \bb < \kappa$, we have $v_\aa \s v_\bb$, and hence $u(v_\aa) > u(v_\bb)$. 
 This shows that $f$ is a strictly increasing function of $\aa$,
and hence there are uncountably many intervals $(f(\aa), f(\aa+1)) \subset \RRR$ that are nonempty, pairwise disjoint, and open.   
But that is impossible: each open interval must contain a rational number, of which there are countably many. 
\end{proof}

\begin{proof}[{\bf Proof of Proposition~\ref{p:unique}}]
Suppose a preorder $\se$ on $M$ is represented by $\UU \subset M^*$. Let $(M^*)^+ \subset M^*$ consist of the functions in $M^*$ that are increasing with respect to $\se$. Write $\KK = \cone{(\UU \cup \CC)}$. To prove the Proposition, it is sufficient to show that  $\ov \KK = (M^*)^+$. 

We first verify $\ov \KK \subset (M^*)^+$. It is obvious that $\KK \subset (M^*)^+$. Suppose $(f_\aa)$ is a net in $\KK$ converging to $f$, and suppose $x \se y$. Then $f_\aa(x) \geq f_\aa(y)$ for all $\aa$. 
Since $M^*$ has the topology of pointwise convergence, 
$\lim_\aa f_\aa(x)=f(x)$ and
$\lim_\aa f_\aa(y)=f(y)$; therefore $f(x)\geq f(y)$. Thus $f$ is increasing, i.e.~$f \in (M^*)^+$. 

Conversely, to show $(M^*)^+ \subset \ov \KK$, 
we first embed $M$ in $M^{**}$, the algebraic dual of $M^*$, 
via the mapping $\phi \colon M \to M^{**}$ given by $\phi(x)(f) = f(x)$. 
It is easy to check that $\phi$ is mixture-preserving (it is also injective, as shown in \citet{pM2001}, but we do not use this). 
The subspace  $\Span(\phi(M)) \subset M^{**}$ 
separates the points of $M^*$, 
so $(M^*, \Span(\phi(M))$ is a dual pair
of vector spaces. Moreover, the topology on $M^*$ is the weak topology with respect to this pairing, so 
it follows from the fundamental theorem of duality 
\citep[Thm. 5.93]{AB2006} that $\Span(\phi(M))$ is the continuous dual of $M^*$.

Suppose for a contradiction that $f \in (M^*)^+$ but $f \notin \ov \KK$. 
The vector space $M^*$ is locally convex, and since $\KK$ is a convex cone, 
we may use the 
strong separating hyperplane theorem 
\citep[Cor. 5.80]{AB2006} to obtain $F \in \Span(\phi(M))$ such that $F(\ov \KK) \subset [0, \infty)$ and $F(f) <0$. 
Write $F = \sum_{x \in M} \ll_x \phi(x) -  \sum_{x \in M} \mm_x \phi(x)$ for nonnegative $\ll_x, \mm_x\in\RRR$, only finitely many nonzero. 
Since $\phi$ is mixture preserving, 
we can combine terms to obtain
$F = \ll \phi(x) - \mm \phi(y)$ for some nonnegative $\ll, \mu\in\RRR$, and $x, y \in M$. 
Since $F$ is nonnegative on $\ov \KK$, and hence on the constant functions, we must have $\ll = \mu$. 
Thus $F(f) = \ll(f(x) - f(y)) < 0$. Since $f$ is increasing, it follows that $x \not\se y$. Thus for some $g \in \UU$,  $g(x) < g(y)$, implying that $F(g) < 0$. 
This is impossible since $g\in\ov\KK$. 
\end{proof}

\bibliographystyle{plainnat}

\end{document}